\documentclass[11pt]{article}
\usepackage[letterpaper,textwidth=6.5in,textheight=9in,
            centering,ignorehead,nomarginpar]{geometry}

\usepackage[T1]{fontenc}
\usepackage{textcomp}
\usepackage{xcolor}

\usepackage{setspace}

\usepackage{url}
\usepackage{hyperref}
\usepackage{doi}

\usepackage{array}
\usepackage{natbib}
\bibpunct{(}{)}{;}{a}{,}{;}
\usepackage{graphicx}
\usepackage{bbm}
\usepackage{amsmath,amsthm,amssymb,amsfonts,latexsym}
\usepackage{mathtools}
\usepackage{booktabs}
\usepackage{icomma}
\usepackage{afterpage}
\usepackage{titling}
\thanksmarkseries{alph}
\usepackage[labelfont={small,sc},textfont={small,it},
            margin=15pt,skip=10pt,position=bottom]{caption}
\usepackage[labelfont={scriptsize,normalfont},textfont={scriptsize,it},
            margin=20pt,skip=5pt,position=bottom,
            labelformat=simple]{subcaption}

\usepackage{algorithmic}
\usepackage[boxed]{algorithm}

\newtheorem{remark}{Remark}

\newtheorem{condition}{Condition}

\newtheorem{proposition}{Proposition}

\numberwithin{condition}{section}
\numberwithin{assumption}{section}
\numberwithin{remark}{section}
\numberwithin{equation}{section}
\numberwithin{lemma}{section}
\numberwithin{definition}{section}
\numberwithin{theorem}{section}
\numberwithin{proposition}{section}
\numberwithin{table}{section}
\numberwithin{figure}{section}
\numberwithin{theorem}{section}
\numberwithin{corollary}{section}
\numberwithin{property}{section}
\numberwithin{algorithm}{section}

\newcommand{\EQ}{\begin{equation}}
\newcommand{\EN}{\end{equation}}
\newcommand{\EQS}{\begin{equation*}}
\newcommand{\ENS}{\end{equation*}}
\newcommand{\ds}{\displaystyle}

\newcommand*\xbar[1]{%
  \hbox{%
    \vbox{%
      \hrule height 0.5pt
      \kern0.5ex%
      \hbox{%
        \kern-0.1em%
        \ensuremath{#1}%
        \kern-0.1em%
      }%
    }%
  }%
}

\def\n1{n}
\def\argmax{\mathop{\rm arg\,max}}

\newcommand{\qq}{\mathfrak{q}}
\newcommand{\pp}{\mathfrak{p}}

\newsavebox{\savepar}

\numberwithin{equation}{section}
\numberwithin{table}{section}
\numberwithin{figure}{section}

\def\argmax{\mathop{\rm arg\,max}}

\usepackage[running]{lineno}

\begin{document}
\title{
Optimal performance of a tontine overlay subject to withdrawal constraints
}

\author{Peter A. Forsyth\thanks{David R. Cheriton School of Computer Science,
        University of Waterloo, Waterloo ON, Canada N2L 3G1,
        \texttt{paforsyt@uwaterloo.ca}, +1 519 888 4567 ext.\ 34415.}
  \and
        Kenneth R. Vetzal\thanks{School of Accounting and Finance,
        University of Waterloo, Waterloo ON, Canada N2L 3G1,
        \texttt{kvetzal@uwaterloo.ca}, +1 519 888 4567 ext.\ 46518.}
     \and
      Graham Westmacott\thanks{Richardson Wealth,
        540 Bingemans Centre Drive, Suite 100,
        Kitchener, ON, Canada N2B 3X9,
        \texttt{Graham.Westmacott@RichardsonWealth.com}, +1 800 451 4748.}
}

\maketitle


\begin{abstract}
We consider the holder of an individual tontine retirement account, with maximum and minimum
withdrawal amounts (per year) specified.  The tontine account holder initiates
the account at age 65, and  earns mortality credits
while alive, but forfeits all wealth in the account upon death.  The holder desires to maximize
total withdrawals, and minimize the expected shortfall, assuming the 
holder survives to age 95. The investor controls the amount
withdrawn each year and the fraction of the investments in stocks and bonds.  The optimal controls
are determined based on a parametric model fitted to almost a century of market data.
The optimal control algorithm is based on dynamic programming and solution of
a partial integro differential equation (PIDE) using Fourier methods.
The optimal strategy (based on the parametric model) is tested out of sample using stationary block bootstrap resampling of the
historical data.  In terms of an expected total withdrawal, expected shortfall (EW-ES)
efficient frontier, the tontine overlay greatly outperforms an optimal strategy
(without the tontine overlay), which in turn outperforms a constant weight strategy
with withdrawals based on the ubiquitous four per cent rule.

\vspace{5pt}

\end{abstract}

\section{Introduction}
It is now commonplace to observe that defined benefit (DB) plans are disappearing.  A recent OECD
study \citep{OECD_2019} observes that less than 50\% of pension assets in 2018 were held in DB
plans in over 80\% of countries reporting.  Of course, the level of assets in defined contribution (DC)
plans is a lagging indicator, since historically, many employees were covered by traditional
DB plans.  These traditional DB plans still have a sizeable share of pension assets, simply because
these plans have accumulated contributions over a longer period of time.

Consider the typical case of a DC plan investor upon retirement.  Assuming that the investor has managed to
accumulate a reasonable amount in her DC plan, the investor now faces the problem of 
determining a decumulation strategy, i.e. how to invest and spend during retirement.
It is often suggested that retirees should purchase annuities, but this is
an unpopular strategy \citep{Peijnenburg2016}.  \citet{MacDonald2013} note that this avoidance
of annuities can be entirely rational. 

A major concern of DC plan investors during the decumulation phase is,
naturally, running out of savings.  Possibly the most widely
cited benchmark strategy is the {\em 4\% rule} \citep{Bengen1994}.  This rule
posits a retiree who invests in a portfolio of 50\% stocks and 50\% bonds, rebalanced
annually, and withdraws 4\% of the original portfolio value each year (adjusted for
inflation).  This strategy would have
never depleted the portfolio over any rolling thirty year historical period tested by Bengen
on US data.
This rule has been revisited many times. For example, \citet{Guyton-Klinger:2006} suggest several
heuristic modifications involving withdrawal amounts and investment strategies.  

Another approach has been suggested by \citet{Waring2015}, which they term an Annually
Recalculated Virtual Annuity (ARVA) strategy.  The idea here is that the amount withdrawn
in any given year should be based on the cash flows from a virtual (i.e. theoretical)
fixed   term  annuity that could be purchased using the existing value of the portfolio.
In this case, the DC plan can never run out of cash, but the withdrawal amounts 
can become arbitrarily small.

Turning to asset allocation strategies, \citet{Irlam:2014} used dynamic programming methods
to conclude that deterministic (i.e. glide path) allocation strategies are sub-optimal.
Of course, the asset allocation strategy and the withdrawal strategy are intimately linked.
A more systematic approach to the decumulation problem involves formulating decumulation
strategies as a problem in optimal stochastic control.  The objective function for this
problem involves a measure of risk and reward, which are, of course, conflicting measures.
\citet{Forsyth_2021_b} uses the withdrawal amount and the asset allocation (fraction
in stocks and bonds) as controls.  The measure of reward is the total (real) accumulated
withdrawal amounts over a thirty year period.  The withdrawal amounts have minimum and
maximum constraints, hence there is a risk of depleting the portfolio.
The measure of risk is the expected shortfall
at the 5\% level, of the  (real) value of the portfolio at the thirty year mark.  
Utilizing both withdrawal amounts and asset allocation as controls considerably reduces
the risk of portfolio depletion compared to fixed allocation or fixed withdrawal strategies.

A recent innovation in retirement planning involves the use of modern tontines 
\citep{Donnelly_2014,Donnelly_2015,Milevsky_2015_a,Fullmer_2019_a,Weinhart_2021,Winter_2020,Milvesky_book}.  
In a tontine,
the investor makes an irrevocable investment in a pooled fund for a fixed time frame.  If the investor
dies during the time horizon of the investment, the investor's portfolio is divided amongst
the remaining (living) members of the fund.  If the investor survives until the end of 
the time horizon, then she will earn mortality credits from those members who have
passed away.  Unlike an annuity, there are no guaranteed cash flows, since typically the
funds are invested in risky assets.  Since there are no guarantees, the expected cash flows
from a tontine are larger than for an annuity (with the same initial investment).
Some authors have argued that the {\em annuity puzzle} should be replaced
by the {\em tontine puzzle}, i.e. since tontines seem to very efficient products
for pooling longevity risk, it is puzzling that the tontine market is
still in its infancy \citep{Chen_2022_x}.
Pooled funds with tontine characteristics have been in use for some
time.\footnote{The variable annuity funds offered by TIAA {\url{https://www.tiaa.org/public/}},
the University of British Columbia pension plan {\url{https://faculty.pensions.ubc.ca/}},
and the Australian Q-super fund {\url{https://qsuper.qld.gov.au/}} can all be viewed as having
tontine characteristics.
However, the Q-super fund takes the approach of 
averaging mortality credits over the entire pool, 
giving age independent mortality credits, which would appear to violate
actuarial fairness. \url{https://i3-invest.com/2021/04/behind-qsupers-retirement-design/}
The Q-super fund is perhaps more properly termed a collective defined contribution (CDC) fund.
CDCs {\url{https://www.ft.com/content/10448b2c-1141-4d2e-943c-70cce2caec52}} have been criticized for
lack of transparency and fairness.}
We should also mention that retail investors may find the concept of a tontine appealing,
simply due to the peer-to-peer model for managing longevity risk, which is also
consistent with the trend towards financial disintermediation.\footnote{See \citet{Benthem_2018}
for an experiment with setting up a tontine using blockchain techniques.}
However, tontines may also require changes to existing legislation in some jurisdictions
\citep{GRI_2021}.  
There have also been suggestions for government management of tontine
accounts \citep{Fullmer_2022_a,Fullmer_2022_c}.
The attractiveness of tontines, from a behavioral finance perspective,
is discussed in \citet{chen_2021}.
For an overview comparison of modern tontines to existing decumulation products, we refer the
reader to \citet{Bar_2022}.

Our focus in this article is on individual tontine accounts \citep{Fullmer_2019_a}, whereby
the investor has full control over the asset allocation in her account.  We also allow the investor
to control the withdrawal amount from the account, subject to maximum and minimum constraints.
Usually it is suggested that withdrawal amounts from a tontine account cannot be increased,
to avoid moral hazard issues.\footnote{An obvious case would be if an investor was given
a medical diagnosis with a high probability of a poor outcome, at which point the investor
would withdraw all remaining funds in her account.}   However, we view the maximum withdrawal as the desired withdrawal,
allowing temporary reductions in withdrawals to minimize  sequence of return risk and probability of ruin.

Consider an investor whose objective function uses reward as measured by total expected accumulated
(real) withdrawals (EW) over a thirty year period.  As a measure of risk, the investor uses
the expected shortfall (ES) of the portfolio at the thirty year point. In this work,
we define the expected shortfall as the mean of the worst 5\% of the outcomes
at year thirty.  The investor's controls
are the amount withdrawn each year, and the allocations to stocks and bonds.
The investor follows an optimal strategy to maximize this objective function.

Alternatively, the investor can use the same objective function, with the same controls,
but this time add a tontine overlay (i.e. the investor is part of a pooled tontine).
The investor has control over the withdrawals (subject of course to the same maximum
and minimum constraints), and the allocation strategy.  

Of course, we expect that the investor who uses the tontine overlay would achieve a better result
than without the overlay, due to the mortality credits earned (we assume that the investor does
not pass away during the thirty year investment horizon).  However, this does not come without
a cost.  If the investor passes away, then her portfolio is forfeited.

Therefore, the investor must be compensated with a sizeable reduction in the risk of portfolio depletion,
compared to the no-tontine overlay case.  The objective of this article is to quantify this reduction,
assuming optimal policies are followed in each case.

More precisely, we consider a 65-year old retiree, who can invest in a portfolio consisting of 
a stock index and a bond index, with yearly withdrawals, and rebalancing.  The investor desires
to maximize the multi-objective function in terms of the risk and reward measures described above,
evaluated at the thirty year horizon (i.e. when the investor is 95).  

We calibrate a parametric stochastic model for real (i.e. inflation adjusted) stock and bond
returns, to almost a century of market data.  We then solve the optimal stochastic control
problem numerically, using dynamic programming.  Robustness of the controls is then tested
using block bootstrap resampling of the historical data.

Our main conclusion is that for a reasonable specification of acceptable tail risk (i.e. expected shortfall),
the expected total cumulative withdrawals (EW) are considerably larger with the tontine overlay,
compared to without the overlay.  This conclusion holds even if the tontine overlay has fees of the
order of 50-100 bps per year. Consequently, if the retiree has no bequest motive, and is primarily concerned
with the risk of depleting her account, then a tontine overlay is an attractive solution.

It is also interesting to note that the optimal control for the withdrawal amount is (to a good approximation)
a bang-bang control, i.e. it is only optimal to withdraw either the maximum or minimum amount in any year.
The allocation control essentially starts off with 40-50\% allocation to stocks. The median allocation control
then rapidly reduces the fraction in equities to a very small amount after $5-10$ years.  The median withdrawal control
starts off at the minimum withdrawal amount, and then rapidly increases withdrawals to the maximum after $2-5$ years.
The precise timing of the switch from minimum withdrawal to maximum withdrawal is simply a function of how much
depletion risk (ES) the investor is prepared to take.

\section{Overview of Individual Tontine Accounts}
\subsection{Intuition}
We give a brief overview of modern tontines in this
section. We restrict attention to the case of an
individual tontine account \citep{Fullmer_2019_a}, which is a constituent
of a perpetual tontine pool.
Consider a pool of $m$ investors, who are alive at time $t_{i-1}$.  Let $v_i^j$
be the balance in the portfolio of  investor $j$ at time $t_i$.  In a tontine, if  investor $j$
participates in a tontine pool in time interval $(t_{i-1}, t_i)$, and investor $j$ dies
in that interval, then her portfolio $v_i^j$ is forfeited  and given to the surviving
members of the pool in the form of mortality credits (gains).  Suppose that the probability
that $j$ dies in $(t_{i-1}, t_i)$ is $q_{i-1}^j$.  
Let the tontine gain (mortality credit) for investor $j$, for the period $(t_{i-1}, t_i)$, paid out at time $t_i$,
be denoted by $c_i^j$.  
The tontine will be a fair game if,
for each player $j$, the expected gain from participating in the tontine is zero,
\begin{eqnarray}
   -q_{i-1}^j v_i^j +( 1 - q_{i-1}^j) c_i^j = 0 ~,
\end{eqnarray}
and solving for the tontine gain $ c_i^j$ gives
\begin{eqnarray}
   c_i^j & = & \overbrace{ \biggl(\frac{ q_{i-1}^j}{ 1 - q_{i-1}^j} }^{Gain~rate} \biggr) v_i^j
       ~. ~\label{tontine_gain}
\end{eqnarray}
For notational convenience, we define the tontine gain rate at $t_i$ for investor $j$ as
\begin{eqnarray}
   ( \mathbb{T}_i^g )^j =  \biggl(\frac{ q_{i-1}^j}{ 1 - q_{i-1}^j} \biggr)  \label{tontine_gain_rate} ~.
\end{eqnarray}
In our optimal control formulation, we will typically drop the superscript $j$ from equation (\ref{tontine_gain_rate}),
\begin{eqnarray}
    \mathbb{T}_i^g  =  \biggl(\frac{ q_{i-1}}{ 1 - q_{i-1}} \biggr)  \label{tontine_gain_rate_special} ~,
\end{eqnarray}
since we will consider a given investor $j$ with conditional mortality probability of $q_{i-1}$ in $(t_{i-1}, t_i)$.

\subsection{Group Gain}\label{group_section}
Consider tontine members $j=1,...,m$ who are are alive at $t_{i-1}$.
Let
\begin{eqnarray}
  {\bf{1}}^j_i & = & \begin{cases}
                      1 & {\text{ Investor }} j {\text{ is alive at }} t_{i-1} {\text{ and alive at }} t_i\\
                      0 & {\text{ Investor }} j {\text{ is alive at }} t_{i-1}  {\text{ and dead at }} t_i\\
                   \end{cases} \nonumber \\
   E[ {\bf{1}}^j_i ] & = & 1 -q_{i-1}^j \label{indicator_def} ~,
\end{eqnarray}
where $E[\cdot] $ is the expectation operator.
Note that the total tontine gain for all members at $t_i$  is
\begin{eqnarray}
{\text{ Total Tontine Gains }} & = & \sum_{j=1}^m  {\bf{1}}^j_i c_i^j  \nonumber \\
                               & = & \sum_{j=1}^m {\bf{1}}^j_i \biggl(\frac{ q_{i-1}^j}{ 1 - q_{i-1}^j}  \biggr) v_i^j ~,
                                  \label{total_gain}
\end{eqnarray}
and that the total amount forfeited by members who have died in $(t_{i-1},t_i)$ is
\begin{eqnarray}
   {\text{ Total Forfeited }} & = &  \sum_{j=1}^m  ( 1 - {\bf{1}}^j ) v_i^j ~. \label{total_forfeited}
\end{eqnarray}
Then 
\begin{eqnarray}
    E[ {\text{ Total Tontine Gains }}] & = & \sum_{j=1}^m q_{i-1}^j v_i^j \nonumber \\
    E[ {\text{ Total Forfeited }}  ] & = &  \sum_{j=1}^m q_{i-1}^j v_i^j ~, \label{equal_in_expectation}
\end{eqnarray}
implying that the expected total tontine gain is balanced by the expected total amount forfeited.

In practice, of course, the expected number of deaths in period $(t_{i-1},t_i)$ may not be equal
to the actual number of deaths.  To compensate for this, a practical implementation method
has been suggested in \citep{Sabin_2016,Denuit_2018,Fullmer_2019,Fullmer_2019_a,Winter_2020}.
Denote the
realized group gain at $t_i$ by $G_i$
\begin{eqnarray}
   G_i & = & \frac{ \sum_{j=1}^m  ( 1 - {\bf{1}}^j_i ) v_i^j }
                  { \sum_{j=1}^m {\bf{1}}^j_i c_i^j} ~. \label{group_gain}
\end{eqnarray}
Observe that the values of ${\bf{1}}^j_i$ in equation (\ref{group_gain}) are the
{\em realized} values, not expectations.
The actual tontine gain (mortality credit) $ \hat{c}_i^j$ earned by investor $j$ (assuming investor $j$ is alive at $t_i$) is then
\begin{eqnarray}
   \hat{c}_i^j & = &  G_i \biggl(\frac{ q_{i-1}^j}{ 1 - q_{i-1}^j}  \biggr) v_i^j 
              ~.
                    \label{actual_credit}
\end{eqnarray}
Essentially, we are scaling equation (\ref{tontine_gain}) by the factor $G_i$, which ensures that the
the total amount forfeited by the observed deaths in $(t_{i-1},t_i )$ is exactly equal to the
total mortality credits disbursed to the survivors.
In the following, we will refer to $\hat{c}_i^j$ as the actual mortality gain,
while ${c}_i^j$ will be termed the nominal tontine gain.

\begin{remark}[$\sum_{j=1}^m {\bf{1}}^j_i c_i^j=0$]
Note that equation (\ref{group_gain}) is undefined if all members die in $(t_{i-1},t_i)$.  We assume that the tontine
is large (in terms of members) and
perpetual, i.e. open to new members, so that the probability of all members dying is negligible.  For mathematical
completeness, we can suppose that if all members die in $(t_{i-1},t_i)$, we collapse the tontine,
and distribute all remaining account values $v_i^j$ to the estates of members $j$.
\end{remark}

While use of equation (\ref{actual_credit}) looks like a reasonable approach,
it turns out that this is not strictly fair in the actuarial sense, i.e. there is
some bias that favors some members over others.
Informally, this bias exists since the total amount forfeited in a period depends on the binary
state of each member of the pool, i.e. alive or dead.\footnote{This is illustrated in
\citet{Winter_2020}, using an example with pool consisting of a large number of young investors (with small
individual portfolios), and
a single elderly member with a large portfolio.  The elderly member effectively subsidizes the younger
members.}
\citet{Sabin_2016} show that the bias is negligible under the following conditions

\begin{condition}[Small bias condition] \label{small_bias_assumption}
Suppose that the following conditions hold
\begin{description}
   \item[(a)] the pool of participants in the tontine is sufficiently large;
   \item[(b)] the expected amount forfeited by all members is large compared to any
         member's nominal gain, i.e.
            \begin{eqnarray}
                \biggl( v_i^j \frac{q_{i-1}^j}{1 - q_{i-1}^j} \biggr) \ll \sum_k q_{i-1}^k v_i^k~~~;~~ j=1,\ldots, m ~,
                      \label{diverse_condition}
             \end{eqnarray}
\end{description}
then the bias is negligibly small \citep{Sabin_2016}.
\end{condition}
Condition (\ref{diverse_condition}) is
essentially a diversification requirement: no member of the pool has an abnormally large share of the
total pool capital.  In addition, of course, if the pool is sufficiently large, then  the actual number of deaths
in $(t_{i-1}, t_i)$ will converge to the expected number of deaths.  

Note that it does not matter what investment strategy is followed by any given investor in period $(t_{i-1}, t_i)$.
Each investor can choose whatever policy they like, since only the observed final portfolio value at $t_i$ matters.
Somewhat counterintuitively, $G_i$ is very close to unity,  even if the participants in the tontine
pool are very heterogeneous, i.e. with different ages, genders, invested amounts, and asset allocations
(assuming Condition \ref{small_bias_assumption} holds).

In \citet{Fullmer_2019}, simulations were carried out to determine the magnitude
of the volatility of $G_i$ under practical sizes of tontine pools.
Given a tontine
pool of 15,000 members, with varying ages, initial capital, and randomly assigned investment
policies (i.e. the bond/stock split), the simulations showed that $E[G_i] \simeq 1$ and
that the standard deviation was about $0.1$.  This standard deviation at each $t_i$ actually
resulted in a smaller effect over a  long term (assuming that the tontine member
lived long enough).  This is simply because everybody dies eventually, so that if fewer
deaths than expected are observed in a year, then more deaths will be observed in later years,
and vice versa.

In the following, we will assume that the pool is sufficiently large and that it
satisfies the diversity condition (\ref{diverse_condition}), so that there is 
no error in assuming that $G \equiv 1$, i.e we will assume that the nominal tontine
gain is the actual tontine gain.
As a sanity check, we also carry out a test whereby we simulate the effect of randomly 
varying $G$, based on the statistics of the simulations in \citet{Fullmer_2019}.  Our results
show that effect of randomness in $G$ can be safely ignored for a reasonably sized tontine pool.
To be more precise here, we will modify equation (\ref{tontine_gain_rate_special}) so that
\begin{eqnarray}
   \mathbb{T}_i^g  =  \biggl(\frac{ q_{i-1}}{ 1 - q_{i-1}} \biggr) G_i  \label{tontine_gain_rate_special_plus_G} ~,
\end{eqnarray}
for a numerical example showing the effects of randomness of $G_i$, in  Monte Carlo
simulations.  Our computation of the optimal strategy will always assume $G_i \equiv 1$.

\subsection{Variable withdrawals}
We will allow the individual tontine member to withdraw variable amounts, subject to
minimum and maximum constraints.  We remind the reader that if a tontine pool is strictly
actuarially fair, then, in theory, there are no constraints on withdrawals and injections
of cash \citep{Braughtigam_2017}.  

However, in practice, since pools
are finite sized, heterogeneous, and mortality credits are not
distributed at infinitesimal intervals, we do not allow arbitrarily large
withdrawals.  This avoids moral hazard issues. 

Since we have a minimum withdrawal amount
in each time period, there is a risk of running out of cash.  We assume that if the tontine
member's account becomes negative, than all trading in this account ceases, and debt accumulates at the
borrowing rate.  
In practice, once the tontine account becomes zero,
the retiree has to fund expenses from another source.
We implicitly assume that the tontine member has other assets which
can be used to fund this minimum consumption level (e.g. real estate).  Of course,
we aim to make this a very improbable event.  In fact, this is the reason why
we allow variable withdrawals.   We can regard the upper bound on the withdrawals
as the desired consumption level, but we allow the tontine member to reduce
(hopefully only temporarily) their withdrawals, to minimize risk of depletion
of their tontine account.

\subsection{Money back guarantees}
In practice, we  observe that many tontine funds offer a money back guarantee.\footnote{ {\url{https://qsuper.qld.gov.au/}}}  
This is usually specified
as a return of the initial (nominal) investment less any withdrawals (if the sum is non-negative) at the time of death.
We do not consider such guarantees in this work, focusing on the pure tontine aspect, which has no guarantees,
and presumably the highest possible expected total withdrawals.  A money back guarantee would have to be hedged,
which would reduce returns.  In practice, this guarantee could be priced separately, and added as overlay
to the tontine investment if desired.

\subsection{Survivor Benefits}
Many DB plans have survivor benefits which are received by a surviving spouse.  A typical case would involve
the surviving spouse receiving $60-75\%$ of the yearly pension after the DB plan holder dies.

Consider the following case of a male, same-sex couple, both of whom are exactly the same age.  As an extreme case, suppose
the survivor benefit is 100\% of the tontine cash flows, which continue until the survivor dies.  From the 
CPM2014 table from the Canadian Institute of Actuaries\footnote{\url{www.cia-ica.ca/docs/default-source/2014/214013e.
pdf}},
the probability that an 85-year old Canadian male dies before reaching the age of 86 is about $.076$.  Assuming that the
mortality probabilities are independent for both spouses, then the probability that both 85-year old spouses die
before reaching age $86$, conditional on both living to age $85$ is $ (.076)^2 \simeq  .0053$.  The tontine
gain rate per year (from equation (\ref{tontine_gain_rate})) is
\begin{eqnarray}
   {\mbox{ tontine gain rate }} & = & \frac{ .0053}{ 1 - .0053} \simeq .0053 ~.
\end{eqnarray}
We will assume in our numerical examples that the base case fee charged for managing the tontine is 
$50$ bps per year.  This means that, net of fees, there are essentially no tontine gains for our hypothetical couple,
for the first 20 years of retirement, which is surely undesirable.  Once one of the partners passes away, the tontine gain rate will,
of course, take a jump in value.

As another extreme case, suppose that the surviving spouse receives $50\%$ of the tontine cash flows.  In this case, 
the total cash flows accruing this couple are exactly the same as dividing the original total wealth into two,
and then having each spouse invest in their own individual tontine.

It is possible to determine the distribution of the cash flows for a survivor benefit which is intermediate
to these edge cases.   However, 
this requires additional state variables in our optimal control problem,
and is probably best tackled using a machine learning approach \citep{yuying_2019,Ni_2022}.
We will leave this case for future work,
and focus attention on the individual tontine case, with no survivor benefit.
Note that in the tontine context, survivor benefits are typically provided by a  separate insurance 
overlay.\footnote{\url{https://i3-invest.com/2021/04/behind-qsupers-retirement-design/}}

\section{Formulation}
We assume that the investor has access to two funds: a broad market stock index fund
and a constant maturity bond index fund.  

The investment horizon is $T$.  Let $S_t$ and $B_t$ respectively denote the 
real (inflation adjusted) \emph{amounts} invested in the
stock index and the bond index respectively.  In general, these
amounts will depend on the investor's strategy over time,
as well as changes
in the real unit prices of the assets.
In the absence of an investor determined
control (i.e. cash withdrawals or rebalancing),
all changes in $S_t$
and $B_t$ result from changes in asset prices. We model the stock index as following
a jump diffusion.  

In addition, we follow the usual practitioner approach and directly model
the returns of the constant maturity bond index as a stochastic process,
see for example \citet{Lin_2015,mitchell_2014}.   
As in \citet{mitchell_2014}, we assume that the constant maturity bond
index follows a jump diffusion process as well.

Let $ S_{t^-} =   S(t - \epsilon), 
\epsilon \rightarrow 0^+$, i.e.\ $t^-$ is the instant of time before
$t$, and let $\xi^s$ be a random number representing a jump multiplier.
When a jump occurs, $S_t = \xi^s S_{t^-}$. 
Allowing for jumps permits modelling of non-normal asset returns.
We assume that $\log(\xi^s)$ follows a double exponential distribution
\citep{kou:2002,Kou2004}. If a jump occurs, $u^s$ is
the probability of an upward jump, while $1-u^s$ is the
chance of a downward jump. The density function for $y = \log (\xi^s)$ is
{\color{black}
\begin{linenomath*}
\begin{equation}
f^s(y) = u^s \eta_1^s e^{-\eta_1^s y} {\bf{1}}_{y \geq 0} +
       (1-u^s) \eta_2^s e^{\eta_2^s y} {\bf{1}}_{y < 0}~.
\label{eq:dist_stock}
\end{equation}
\end{linenomath*}
}
We also define
{\color{black}
\begin{linenomath*}
\begin{eqnarray}
\gamma^s_{\xi} &= &E[ \xi^s -1 ]
          =   \frac{u^s \eta_1^s}{\eta_1^s - 1} + 
          \frac{ ( 1 - u^s ) \eta_2^s }{\eta_2^s + 1}  -1 ~.
\end{eqnarray}
\end{linenomath*}
}
In the absence of control, $S_t$ evolves according to
\begin{eqnarray}
\frac{dS_t}{S_{t^-}} &= &\left(\mu^s -\lambda_\xi^s \gamma_{\xi}^s \right) \, dt + 
  \sigma^s \, d Z^s +  d\left( \ds \sum_{i=1}^{\pi_t^s} (\xi_i^s -1) \right) ,
\label{jump_process_stock}
\end{eqnarray}
where $\mu^s$ is the (uncompensated) drift rate, $\sigma^s$ is the volatility,
$d Z^s$ is the increment of a Wiener process,
$\pi_t^s$ is a Poisson process with positive intensity parameter
$\lambda_\xi^s$, and $\xi_i^s$ are i.i.d.\ positive random variables having
distribution (\ref{eq:dist_stock}).
Moreover, $\xi_i^s$, $\pi_t^s$, and $Z^s$ are assumed to all be
mutually independent.

Similarly,  let the amount in the bond index be $B_{t^-} =  B(t - \epsilon), \epsilon \rightarrow 0^+$.
In the absence of control, $B_t$ evolves as
\begin{eqnarray}
\frac{dB_t}{B_{t^-}} &= &\left(\mu^b -\lambda_\xi^b \gamma_{\xi}^b  
   + \mu_c^b {\bf{1}}_{\{B_{t^-} < 0\}}  \right) \, dt + 
  \sigma^b \, d Z^b +  d\left( \ds \sum_{i=1}^{\pi_t^b} (\xi_i^b -1) \right) ,
\label{jump_process_bond}
\end{eqnarray}
where the terms in equation (\ref{jump_process_bond}) are defined analogously to
equation (\ref{jump_process_stock}).  In particular, $\pi_t^b$ 
is a Poisson process with positive intensity parameter
$\lambda_\xi^b$, and $\xi_i^b$ has distribution 
\begin{linenomath*}
\begin{equation}
f^b( y= \log \xi^b) = u^b \eta_1^b e^{-\eta_1^b y} {\bf{1}}_{y \geq 0} +
       (1-u^b) \eta_2^b e^{\eta_2^b y} {\bf{1}}_{y < 0}~,
\label{eq:dist_bond}
\end{equation}
\end{linenomath*}
and $\gamma_{\xi}^b = E[ \xi^b -1 ]$.  $\xi_i^b$, $\pi_t^b$, and $Z^b$ are assumed to all be
mutually independent.  The term $\mu_c^b {\bf{1}}_{\{B_{t^-} < 0\}}$ in equation 
(\ref{jump_process_bond}) represents the extra cost of borrowing (the spread).

The diffusion processes are correlated, i.e. $d Z^s \cdot d Z^b = \rho_{sb} dt$.  The stock
and bond jump processes are assumed mutually independent.
See \citet{forsyth_2020_a} for justification of the assumption of stock-bond jump independence.

We define the investor's total wealth at time $t$ as
\begin{equation}
\text{Total wealth } \equiv W_t = S_t + B_t.
\end{equation}
We impose the constraints that (assuming solvency) shorting stock and using leverage
(i.e.\ borrowing) are not permitted.
In the event of insolvency (due to withdrawals), the portfolio
is liquidated, trading ceases and debt accumulates
at the borrowing rate.

\section{Notational Conventions}
\label{adaptive_section}
Consider a set of discrete withdrawal/rebalancing times $\mathcal{T}$
\begin{eqnarray}
   \mathcal{T} = \{t_0=0 <t_1 <t_2< \ldots <t_M=T\}  \label{T_def}
\end{eqnarray}
where we assume that $t_i - t_{i-1} = \Delta t =T/M$ is constant for
simplicity.
To avoid subscript clutter, in the following, we will
occasionally use the notation $S_t \equiv S(t), B_t \equiv B(t)$ and
$W_t \equiv W(t)$.
Let the inception time of the investment be $t_0 = 0$. We let
$\mathcal{T} $ be the set of
withdrawal/rebalancing times, as defined in equation (\ref{T_def}).
At each rebalancing time
$t_i$, $i = 0, 1, \ldots, M-1$, the investor (i)~withdraws an amount of cash
$\qq_i$ from the portfolio, and then~(ii) rebalances the portfolio. At
$t_M = T$,  the portfolio is liquidated and no cash flow occurs.  For notational completeness,
this is enforced by specifying $\qq_M = 0$.

In the following, given a time dependent function $f(t)$, then we will use
the shorthand notation
\begin{eqnarray}
  f(t_i^+) \equiv \displaystyle  \lim_{\epsilon \rightarrow 0^+}
          f(t_i + \epsilon) ~~&; & ~~
      f(t_i^-) \equiv \displaystyle  \lim_{\epsilon \rightarrow 0^+}
          f(t_i - \epsilon)  ~~.
\end{eqnarray}
Let 
\begin{eqnarray}
   (\Delta t)_i & = \begin{cases}
                    \Delta t &  i = 1, \ldots M, \\
                     0    &         i=0
                \end{cases} ~~.
\end{eqnarray}
We assume that a tontine fee of $\mathbb{T}^f$ per unit time is charged at $t \in \mathcal{T}$,
based on the total portfolio value at $t_i^-$,
after tontine gains but before withdrawals.
Recalling the definition of tontine gain rate $\mathbb{T}^g_i$ in equation (\ref{tontine_gain_rate_special}),
we modify this definition to enforce no tontine gain at $t=0$,
\begin{eqnarray}
   \mathbb{T}_i^g  =  \begin{cases}
                       \biggl(\frac{ q_{i-1}}{ 1 - q_{i-1}} \biggr) & i=1, \ldots, M \\
                       0 & i = 0\\
                     \end{cases} ~.
         \label{T_g_modified}
\end{eqnarray}
Then, $W(t_i^+)$ is given by
\begin{eqnarray}
   W(t_i^+) = \biggl( S(t^-_i) + B(t_i^-) \biggr) \biggl( 1 + \mathbb{T}^g_i \biggr)  \exp( -(\Delta t)_i  \mathbb{T}^f ) -\mathfrak{q}_i  ~;~ i \in \mathcal{T}
   ~,  \label{W_plus_def}
\end{eqnarray}
where we recall that  $\mathfrak{q}_M \equiv 0$ and $ (\Delta t)_0 \equiv 0$.  
With some abuse of notation, we define
\begin{eqnarray}
   W(t_i^-) & = & \biggl( S(t^-_i) + B(t_i^-) \biggr) 
                \biggl( 1 + \mathbb{T}^g_i \biggr)  \exp( -(\Delta t)_i \mathbb{T}^f )
     \label{W_minus_def}
\end{eqnarray}
as the total portfolio value, after tontine gains and tontine fees, the instant before withdrawals and
rebalancing at $t_i$.

Typically, DC plan savings are held in a tax advantaged
account, with no taxes triggered by rebalancing.  With infrequent (e.g. yearly) rebalancing, we also
expect other transaction costs, apart from the tontine fees, to be small, and hence can be ignored.  It is possible to include
transaction costs, but at the expense of increased computational cost \citep{Van2018}.

We denote by
$X\left(t\right)=
\left( S \left( t \right), B\left( t \right) \right)$,
$t\in\left[0,T\right]$, the multi-dimensional controlled
underlying process, and by $x = (s, b)$
the realized state of the system.
Let the rebalancing control $\pp_i(\cdot)$ be  the fraction invested in the stock index
at the rebalancing date $t_i$, i.e.
\begin{eqnarray}
    \pp_i \left( X(t_i^-) \right) = 
       \pp \left( X(t_i^-), t_i \right) & = & \frac{ S(t_i^+)} {S(t_i^+) + B(t_i^+) } ~.
\end{eqnarray}

Let the withdrawal control $\qq_i(\cdot)$ be the amount withdrawn at time $t_i$, i.e.
$\qq_i \left( X(t_i^-) \right)  = \qq  \left( X(t_i^-), t_i \right)$.
Formally,  the controls depend on the state
of the investment portfolio, before the rebalancing
occurs, i.e.
$\pp_i(\cdot) =  \pp\left(X(t_i^-), t_i)\right) 
= \pp\left(X_i^-, t_i \right)$, and
$\qq_i(\cdot) =  \qq\left(X(t_i^-), t_i)\right) 
= \qq\left(X_i^-, t_i \right)$,
$t_i \in \mathcal{T}$, where $\mathcal{T}$ is the
set of rebalancing times.

However, it will be convenient to note that
in our case, we find the optimal control $\pp_i(\cdot)$
amongst all strategies with constant wealth (after withdrawal of cash).
Hence, with some abuse of notation, we will now consider
$\pp_i(\cdot)$ to be function of wealth after withdrawal of cash
\begin{linenomath*}
\begin{eqnarray}
  \pp_i(\cdot) &= & \pp(W(t_i^+), t_i) \nonumber \\
       W(t_i^+) &= & W(t_i^-) - \qq_i(\cdot) \nonumber \\
        W(t_i^-) &=& \biggl( S(t^-_i) + B(t_i^-) \biggr) \biggl( 1 + \mathbb{T}^g_i \biggr)  \exp( -(\Delta t)_i  \mathbb{T}^f )
                                       \nonumber \\
        S(t_i^+) &= &S_i^+ = \pp_i(W_i^+)~ W_i^+ \nonumber \\
 B(t_i^+) & = &B_i^+ =   (1 -\pp_i(W_i^+)) ~W_i^+
       ~~.  \label{p_def_2}
\end{eqnarray}
\end{linenomath*}
Note that the control for $\pp_i(\cdot)$ depends only $W_i^+$.
Since $\pp_i(\cdot) = \pp_i(W_i^- - \qq_i)$, then it follows that
\begin{eqnarray}
  \qq_i(\cdot) = \qq_i(W_i^-) \label{q_dependence}~,
\end{eqnarray}
which we discuss further in Section \ref{DP_program}.

A control at time $t_i$, is then given by the pair $( \qq_i(\cdot), \pp_i(\cdot) )$ where the notation $(\cdot)$
denotes that the control is a function of the state.

Let $\mathcal{Z}$ represent the set of admissible
values of the controls $(\qq_i(\cdot), \pp_i(\cdot))$.
We impose no-shorting, no-leverage
constraints (assuming solvency).  We also impose maximum and minimum values for the withdrawals.
We apply the constraint that in the event of insolvency due to withdrawals ($W(t_i^+) < 0$),
trading ceases and debt (negative wealth) accumulates at the appropriate
borrowing rate of return (i.e. a  spread over the bond rate).  
We also specify that the stock assets are liquidated at $t=t_M$.

More precisely, let $W_i^+$ be the wealth after withdrawal of cash, and $W_i^-$ be the total wealth
before withdrawals (but after fees and tontine cash flows), then
define
\begin{eqnarray}
 \mathcal{Z}_{\qq}  & = &
              \begin{cases}
                  [\qq_{\min},  \qq_{\max}  ] &   t \in \mathcal{T} ~;~ t \neq t_M~;~ W_i^- \geq \qq_{\max} \\
                  [\qq_{\min},  \max(\qq_{\min}, W_i^- )    ] &   t \in \mathcal{T} ~;~ t \neq t_M ~;~ W_i^- < \qq_{\max} \\
                  \{0\}   &    t = t_M
             \end{cases}
                      ~,  \label{Z_q_def}\\
    \mathcal{Z}_\pp (W_i^+,t_i) &=&
          \begin{cases}
                  [0,1] & W_i^+ > 0 ~;~ t_i \in \mathcal{T}~;~ t_i \neq t_M \\
                  \{0\} & W_i^+ \leq 0 ~;~ t_i \in \mathcal{T}~;~  t_i \neq t_M \\
                  \{0\} &  t_i=t_M
          \end{cases}    ~.  \label{Z_p_def} \\
\end{eqnarray}
The rather complicated expression in equation (\ref{Z_q_def}) imposes the assumption that, as wealth
becomes small, the retiree (i) tries to avoid insolvency as much as possible and (ii) in the event of
insolvency, withdraws only $\qq_{\min}$.

The set of admissible values for $(\qq_i,\pp_i), t_i \in \mathcal{T}$,
can then be written as
\begin{eqnarray}
   (\qq_i, \pp_i) \in \mathcal{Z}(W_i^-, W_i^+,t_i) & = & \mathcal{Z}_{\qq}(W_i^-, t_i) \times \mathcal{Z}_{\pp} (W_i^+,t_i)~.
  \label{admiss_set}
\end{eqnarray}
For implementation purposes, we have written equation (\ref{admiss_set}) in terms of the wealth after
withdrawal of cash.  However, we remind the reader that since $W_i^+ = W_i^- -\qq_i$, the controls
are formally a function of the state $X(t_i^-)$ before the control is applied.

The admissible control set $\mathcal{A}$ can then be written as
\begin{eqnarray}
  \mathcal{A} = \biggl\{
                  (\qq_i, \pp_i)_{0 \leq i \leq M} : (\pp_i, \qq_i) \in \mathcal{Z}(W_i^-, W_i^+,t_i) 
                \biggr\}
\end{eqnarray}
An admissible control $\mathcal{P} \in \mathcal{A}$, where $\mathcal{A}$ is
the admissible control set, can be written as,
\begin{eqnarray}
    \mathcal{P} = \{ (\qq_i(\cdot), \pp_i(\cdot) ) ~:~ i=0, \ldots, M \} ~.
\end{eqnarray}
We also define $\mathcal{P}_n \equiv \mathcal{P}_{t_n} \subset \mathcal{P}$
as the tail of the set of controls in $[t_n, t_{n+1}, \ldots, t_{M}]$, i.e.
\begin{eqnarray}
   \mathcal{P}_n =\{ (   \qq_n(\cdot) , \pp_n(\cdot) ),  \ldots, 
                           (\qq_{M}(\cdot) , \pp_{M}(\cdot) ) \} ~.
\end{eqnarray}
For notational completeness, we also define the tail of the admissible control set $\mathcal{A}_n$ as
\begin{eqnarray}
   \mathcal{A}_n = \biggl\{
                  ( \qq_i, \pp_i )_{n \leq i \leq M} : (\qq_i, \pp_i) \in \mathcal{Z}(W_i^-, W_i^+,t_i) 
                \biggr\}
\end{eqnarray}
so that $\mathcal{P}_n \in \mathcal{A}_n $.

\section{Risk and Reward}
\subsection{Risk: Definition of Expected Shortfall (ES)}
Let $g(W_T)$ be the probability density function of wealth $W_T$ at $t=T$.
Suppose
\begin{linenomath*}
\begin{equation}
\int_{-\infty}^{W^*_{\alpha}}  g(W_T) ~dW_T = \alpha,
    \label{CVAR_def_a}
\end{equation}
\end{linenomath*}
i.e.\ \emph{Pr}$[W_T >  W^*_{\alpha}] = 1 - \alpha$.  We can interpret
$W^*_{\alpha}$ as the Value at Risk (VAR) at level $\alpha$\footnote{In practice, the negative of $W^*_{\alpha}$ is often the
reported VAR.}.
The Expected Shortfall (ES) at level $\alpha$ is then
\begin{linenomath*}
\begin{equation}
{\text{ES}}_{\alpha} = \frac{\int_{-\infty}^{W^*_{\alpha}} W_T ~  g(W_T) ~dW_T }
                      {\alpha},
\label{ES_def_1}
\end{equation}
\end{linenomath*}
which is the mean of the worst $\alpha$ fraction of outcomes.
Typically $\alpha \in \{ .01, .05 \}$. The definition of ES in
equation \eqref{ES_def_1} uses the probability density of the final
wealth distribution, not the density of \emph{loss}. Hence, in our case,
a larger value of ES (i.e.\ a larger value of average worst case terminal
wealth) is desired.  The negative of ES is commonly referred to as Conditional
Value at Risk (CVAR).

Define $X_0^+ = X(t_0^+), X_0^- = X(t_0^-)$.  Given an expectation under control $\mathcal{P}$, $E_{\mathcal{P}}
[\cdot]$, as noted by \citet{Uryasev_2000},
$ \text{ES}_{\alpha}$ can be alternatively written as
\begin{eqnarray}
  {\text{ES}}_{\alpha}( X_0^-, t_0^-)  & = & 
      \sup_{W^*} E_{\mathcal{P}_0}^{X_0^+, t_0^+}
   \biggl[W^* + \frac{1}{\alpha} \min(  W_T - W^* , 0 )
                       \biggr]  ~.
\label{ES_def}
\end{eqnarray}
The admissible set for $W^*$ in equation (\ref{ES_def}) is over
the set of possible values for $W_T$.  

The notation ${\text{ES}}_{\alpha}( X_0^-, t_0^-) $
emphasizes that  ${\text{ES}}_{\alpha}$ is as seen at  $(X_0^-, t_0^-)$.  In other words,
this is the pre-commitment ${\text{ES}}_{\alpha}$.  A strategy based purely on optimizing
the pre-commitment value of ${\text{ES}}_{\alpha}$ at time zero  is {\em time-inconsistent},
hence has been termed by many as {\em non-implementable}, since the investor
has an incentive to deviate from the time zero pre-commitment strategy at $t >0$.
However, in the following, we will consider the pre-commitment strategy merely
as a device to determine an appropriate level of $W^*$ in equation (\ref{ES_def}).
If we fix $W^*$ $\forall t >0$, then this strategy is the induced time consistent
strategy \citep{Strub_2019_a}, hence is implementable.  We delay further discussion of this subtle point
to Appendix \ref{appendix_induced}.

\subsection{A measure of reward: expected total withdrawals (EW)}
We will use expected total withdrawals as a measure of reward in the following. More precisely, we define EW (expected withdrawals) as
\begin{eqnarray}
      {\text{EW}}( X_0^-, t_0^-) = E_{\mathcal{P}_0}^{X_0^+, t_0^+} 
                                 \biggl[ 
                                    \sum_{i=0}^{M} q_i
                                 \biggr] ~.  \label{EW_def}
\end{eqnarray}
Note that there is no discounting term in equation (\ref{EW_def}) (recall that all
quantities are real, i.e. inflation adjusted). It is straightforward to
introduce discounting, but we view setting the real discount rate to zero to
be a reasonable and conservative choice.  See \citet{Forsyth_2021_b} for further comments.

\section{Problem EW-ES}
Since expected withdrawals (EW) and expected shortfall (ES) are conflicting measures,
we use a scalarization technique to find the Pareto points for this multi-objective
optimization problem.  Informally, for a given scalarization parameter $\kappa >0$,
we seek to find the control $\mathcal{P}_0$ that maximizes
\begin{eqnarray}
 {\text{EW}}( X_0^-, t_0^-)  + \kappa~ {\text{ES}}_{\alpha}( X_0^-, t_0^-)
     ~.  \label{objective_overview}
\end{eqnarray}
More precisely, we define the pre-commitment EW-ES problem $(PCES_{t_0}(\kappa))$
problem in terms of the value function $J(s,b,t_0^-)$ 

\begin{eqnarray}
\left(\mathit{PCES_{t_0}}\left(\kappa \right)\right):
    \qquad J\left(s,b,  t_{0}^-\right)  
    & = & \sup_{\mathcal{P}_{0}\in\mathcal{A}}
          \sup_{W^*}
        \Biggl\{
               E_{\mathcal{P}_{0}}^{X_0^+,t_{0}^+}
           \Biggl[ ~\sum_{i=0}^{M} \qq_i ~  + ~
              \kappa \biggl( W^* + \frac{1}{\alpha} \min (W_T -W^*, 0) \biggr)
                    \Biggr. \Biggr. \nonumber \\
         & &  \Biggl. \Biggl. ~~~~~
               + \epsilon W_T 
                \bigg\vert X(t_0^-) = (s,b)
                   ~\Biggr] \Biggr\}\label{PCES_a}\\
    \text{ subject to } & &
               \begin{cases}
(S_t, B_t) \text{ follow processes \eqref{jump_process_stock} and \eqref{jump_process_bond}};  
 ~~t \notin \mathcal{T} \\
      W_{\ell}^+ = W_{\ell}^{-}  - \qq_\ell \,; ~ X_\ell^+ = (S_\ell^+ , B_\ell^+)  \\
      W_{\ell}^{-} = \biggl( S(t^-_i) + B(t_i^-) \biggr) \biggl( 1 + \mathbb{T}^g_i \biggr)  
                        \exp( -(\Delta t)_i  \mathbb{T}^f )\\
   S_\ell^+ = \pp_\ell(\cdot) W_\ell^+ \,; 
 ~B_\ell^+ = (1 - \pp_\ell(\cdot) ) W_\ell^+ \, \\
    ( \qq_\ell(\cdot) , \pp_\ell(\cdot) )  \in \mathcal{Z}(W_\ell^-, W_\ell^+,t_\ell)  \\
    \ell = 0, \ldots, M ~;~ t_\ell \in \mathcal{T}  \\
               \end{cases}  ~.
\label{PCES_b}
\end{eqnarray}

Note that we have added the extra term $E_{\mathcal{P}_{0}}^{X_0^+,t_{0}^+}[ \epsilon W_T ]$ to
equation (\ref{PCES_a}).  If we have a maximum withdrawal constraint, and if $W_t \gg W^*$ as
$t \rightarrow T$, the controls become ill-posed.  In this fortunate state for the investor,
we can break the investment policy ties by setting $\epsilon < 0$, which will force investments
in bonds, while if $\epsilon > 0$, then this will force investments into stocks.   Choosing
$ |\epsilon| \ll 1$ ensures that this term only has an effect if $W_t \gg W^*$ and $t \rightarrow T$.
See \citet{Forsyth_2021_b} for more discussion of this.

Interchange the $\sup \sup (\cdot) $ in equation (\ref{PCES_a}), so that
value function $ J\left(s,b,  t_{0}^-\right)$ can be
written as
\begin{eqnarray}
\qquad J\left(s,b,  t_{0}^-\right)
     & = &  \sup_{W^*} \sup_{\mathcal{P}_{0}\in\mathcal{A}}
              \Biggl\{
               E_{\mathcal{P}_{0}}^{X_0^+,t_{0}^+}
           \Biggl[ ~\sum_{i=0}^{M} \mathfrak{q}_i ~  + ~
              \kappa \biggl( W^* + \frac{1}{\alpha} \min (W_T -W^*, 0) \biggr)
                 + \epsilon W_T \bigg\vert X(t_0^-) = (s,b)
                   ~\Biggr] \Biggr\} ~.  \nonumber \\
             \label{pcee_inter}
\end{eqnarray}
Noting that the inner supremum in equation (\ref{pcee_inter}) is a continuous
function of $W^*$, and noting that the optimal value of $W^*$ in equation (\ref{pcee_inter})
is bounded\footnote{This is the same as noting that a finite value at risk exists.  This easily
shown,  assuming $ 0 <\alpha < 1$, since our investment strategy uses no leverage and no-shorting.}, then define
\begin{eqnarray}
\mathcal{W}^*(s,b)  & = & \displaystyle \argmax_{W^*} \biggl\{
                      \sup_{\mathcal{P}_{0}\in\mathcal{A}} 
            \Biggl\{
               E_{\mathcal{P}_{0}}^{X_0^+,t_{0}^+}
           \Biggl[ ~\sum_{i=0}^{M} \mathfrak{q}_i ~  + ~
              \kappa \biggl( W^* + \frac{1}{\alpha} \min (W_T -W^*, 0) \biggr) 
                    \biggr. \biggr. \nonumber \\
               & &  \biggl. \biggl. ~~~~~~~~~~~~ + \epsilon W_T ~ \bigg\vert X(t_0^-) = (s,b)
                   ~\Biggr] \Biggr\} ~.
             \nonumber \\
            \label{pcee_argmax}
\end{eqnarray}
We refer the reader to \citet{forsyth_2019_c} for an extensive discussion
concerning pre-commitment and time consistent ES strategies.  We summarize
the relevant results from that discussion in Appendix \ref{appendix_induced}.

\section{Formulation as a Dynamic Program}\label{DP_program}

We use the method in \citet{forsyth_2019_c} to solve problem (\ref{pcee_inter}).
We write equation (\ref{pcee_inter}) as
\begin{eqnarray}
   J(s,b,t_0^-) & = & \sup_{W^*} V(s,b,W^*, 0^-)~,
\end{eqnarray}
where the auxiliary function $V(s, b, W^*, t)$
is defined as
\begin{eqnarray}
   V(s, b, W^*, t_n^-) & = & \sup_{\mathcal{P}_{n}\in\mathcal{A}_n}
        \Biggl\{
               E_{\mathcal{P}_{n}}^{\hat{X}_n^+,t_{n}^+}
           \Biggl[
                \sum_{i=n}^{M} \qq_i +
           \kappa
             \biggl(
                  W^* + \frac{1}{\alpha} \min((W_T -W^*),0) 
               \biggr) + \epsilon W_T ~
              \bigg\vert \hat{X}(t_n^-) = (s,b, W^*)  \Biggr]
                   \Biggr\}~. \nonumber \\
                    ~ \label{expanded_1} \\
     \text{ subject to } & &
  \begin{cases}
(S_t, B_t) \text{ follow processes \eqref{jump_process_stock} and \eqref{jump_process_bond}};  
 ~~t \notin \mathcal{T} \\
      W_{\ell}^+ = W_{\ell}^-  - \qq_\ell \,; ~ X_\ell^+ = (S_\ell^+ , B_\ell^+, W^*)  \\
       W_{\ell}^- = \biggl( S(t^-_i) + B(t_i^-) \biggr) \biggl( 1 + \mathbb{T}^g_i \biggr)  
                           \exp( -(\Delta t)_i  \mathbb{T}^f )\\
   S_\ell^+ = \pp_\ell(\cdot) W_\ell^+ \,; 
 ~B_\ell^+ = (1 - \pp_\ell(\cdot) ) W_\ell^+ \, \\
    ( \qq_\ell(\cdot), \pp_\ell(\cdot) )  \in \mathcal{Z}(W_\ell^-, W_\ell^+ ,t_\ell)   \\
    \ell = n, \ldots, M ~;~ t_\ell \in \mathcal{T}  \\
               \end{cases}  ~.
             \label{expanded_2}
\end{eqnarray}

We have now decomposed the original problem (\ref{pcee_inter}) into two steps
\begin{itemize}
   \item For given initial cash $W_0$, and a fixed value of $W^*$, solve problem (\ref{expanded_1})
         using dynamic programming
           to determine $V(0,W_0, W^*, 0^-)$.
   \item Solve problem (\ref{pcee_inter}) by maximizing over $W^*$
         \begin{eqnarray}
      J(0,W_0, 0^-) & = & \sup_{W^*} V(0,W_0, W^*, 0^-) ~.
              \label{final_step_EWES}
       \end{eqnarray}
\end{itemize}

\subsection{Dynamic Programming Solution of Problem (\ref{expanded_1})}
We give a brief overview of the method described in detail in
\citep{Forsyth_2021_b}.
Apply the dynamic programming principle to $t_n \in \mathcal{T}$
\begin{eqnarray}
     V(s,b,W^*, t_n^-) & = & 
                           \sup_{\qq \in \mathcal{Z}_\qq(w^-,t_n) } \biggl\{  ~~~\sup_{\pp \in \mathcal{Z}_\pp(w^- - q, t_n) } 
                         \biggl[  \qq+
                                    V(  (w^- -\qq) \pp ,  (w^- - \qq) (1-\pp) , W^*, t_n^+)
                                     \biggr]
                                                 \biggr\} \nonumber \\
                       & = & \sup_{ \qq \in Z_\qq(w^-, t_n)}  \biggl\{ \qq + 
                                   \biggl[ \sup_{ \pp \in \mathcal{Z}_\pp( w^- - q, t_n)}
                                    V(  (w^- -\qq) \pp ,  (w^- - \qq) (1-\pp) , W^*, t_n^+)
                                     \biggr]
                                                 \biggr\} \nonumber \\
                        & & w^- = (s+b) \biggl( 1 + \mathbb{T}^g_i \biggr)  
                           \exp( -(\Delta t)_i  \mathbb{T}^f )  ~ . \label{dynamic_a}
\end{eqnarray}
For computational purposes, we define
\begin{eqnarray}
   \tilde{V}( w, t_n, W^*) & = & \biggl[ \sup_{ \pp \in \mathcal{Z}_\pp( w, t_n)}
                                    V(  w \pp ,  w (1-\pp) , W^*, t_n^+)
                                     \biggr] ~. \label{dynamic_b}
\end{eqnarray}
Equation (\ref{dynamic_a}) now becomes
\begin{eqnarray}
     V(s,b,W^*, t_n^-) & = & 
                           \sup_{\qq \in \mathcal{Z}_\qq(w^-,t_n)}
                               \biggl\{ \qq+
                              \biggl[ 
                                    \tilde{V}(  (w^- -\qq) , W^*, t_n^+)
                                     \biggr]
                              \biggr\} \nonumber \\
               & &      w^- = (s+b) \biggl( 1 + \mathbb{T}^g_i \biggr)  
                           \exp( -(\Delta t)_i  \mathbb{T}^f )  ~ .
     \label{dynamic_c}
\end{eqnarray}
This approach effectively replaces a two dimensional optimization for $(\qq_n, \pp_n)$,
to two sequential one dimensional optimizations.
From equations (\ref{dynamic_b}-\ref{dynamic_c}), it is clear that the optimal pair $(\qq_n, \pp_n)$
is such that 
\begin{eqnarray}
  \qq_n & =  & \qq_n(  w^-  , W^*) \nonumber \\
     & & w^- = (s+b) \biggl( 1 + \mathbb{T}^g_i \biggr)  
                           \exp( -(\Delta t)_i  \mathbb{T}^f ) \nonumber \\
     \pp_n &= &\pp_n( w, W^*) \nonumber \\ 
          & &  w = w^- - \qq_n ~. \label{dynamic_d}
\end{eqnarray}
In other words, the optimal withdrawal control $ \qq_n$ is only a function of total wealth (after tontine gains and fees)
before withdrawals.
The optimal control $\pp_n$ is a function only of total wealth
after withdrawals, tontine gains, and fees.

At $t=T$, we have
\begin{eqnarray}
V(s, b, W^*,T^+) & = &  \kappa \biggl( 
                               W^* + \frac{\min( (s+b -W^*), 0 )}{\alpha} 
                               \biggr) + \epsilon (s+b) ~. \label{dynamic_b_1}
\end{eqnarray}
At points in between rebalancing times, i.e. $t \notin \mathcal{T}$, the usual arguments (from
SDEs (\ref{jump_process_stock}-\ref{jump_process_bond}), and  
\citet{Forsyth_2021_b}) give
\begin{eqnarray}
  & & V_t +  \frac{ (\sigma^s)^2 s^2}{2} V_{ss} +( \mu^s - \lambda_{\xi}^s \gamma_{\xi}^s) s V_s
       + \lambda_{\xi}^s \int_{-\infty}^{+\infty} V( e^ys, b, t) f^s(y)~dy 
      + \frac{ (\sigma^b)^2 b^2}{2} V_{bb} 
                  \nonumber \\
     & & ~~~     + ( \mu^b + \mu_c^b {\bf{1}}_{ \{ b < 0 \} } - \lambda_{\xi}^b \gamma_{\xi}^b) b V_b       
                  + \lambda_{\xi}^b \int_{-\infty}^{+\infty} V( s, e^yb, t) f^b(y)~dy 
       -( \lambda_{\xi}^s + \lambda_{\xi}^b )  V + \rho_{sb} \sigma^s \sigma^b s b V_{sb} 
       =0  ~,  \nonumber \\
  & & ~~~~~~~~~~~~~~~~~~~~~~~~~~~~~~~s \ge 0 ~; b \ge 0~~ ~.
      \label{expanded_7}
\end{eqnarray}
In case of insolvency\footnote{Insolvency can only occur due to the minimum withdrawals specified.} $s = 0, b < 0$
\begin{eqnarray}
   & & V_t +  
       \frac{ (\sigma^b)^2 b^2}{2} V_{bb} 
       + ( \mu^b + \mu_c^b {\bf{1}}_{ \{ b < 0 \} } - \lambda_{\xi}^b \gamma_{\xi}^b) b V_b       
                  + \lambda_{\xi}^b \int_{-\infty}^{+\infty} V( 0, e^yb, t) f^b(y)~dy 
       - \lambda_{\xi}^b V  
       =0 ~,\nonumber \\ 
  & & ~~~~~~~~~~~~~~~~~~~~~~~~s = 0 ~; b < 0~~ ~.
      \label{expanded_8}
\end{eqnarray}

It will be convenient, for computational purposes, to re-write equation (\ref{expanded_8}) in terms of
debt $\hat{b} = -b$ when $b<0$.
Now let $\hat{V} (\hat{b}, t) = V(0, b,t), b < 0, b = -\hat{b}$ in equation (\ref{expanded_8}) to give
\begin{eqnarray}
 & & \hat{V}_t +  
       \frac{ (\sigma^b)^2 \hat{b}^2}{2} \hat{V}_{\hat{b} \hat{b}} 
       + ( \mu^b + \mu_c^b  - \lambda_{\xi}^b \gamma_{\xi}^b)  \hat{b} \hat{V}_{\hat{b}}       
                  + \lambda_{\xi}^b \int_{-\infty}^{+\infty} \hat{V}(  e^y \hat{b}, t) f^b(y)~dy 
       - \lambda_{\xi}^b \hat{V}  
       =0 ~,\nonumber \\ 
  & & ~~~~~~~~~~~~~~~~~~~~~~~~s = 0 ~;~b < 0~;~  \hat{b}  = -b ~.
      \label{minus_b_2}
\end{eqnarray}
Note that equation (\ref{minus_b_2}) is now amenable to a transformation of the form $\hat{x} = \log \hat{b}$
since $\hat{b} > 0$, which is required when using a Fourier method \citep{forsythlabahn2017,Forsyth_2021_b}
to solve equation (\ref{minus_b_2}).

After rebalancing, if $b \ge 0$, then $b$ cannot become negative, since $b=0$ is a barrier in
equation (\ref{expanded_8}).  However, $b$ can become negative after withdrawals, in which case
$b$ remains in the state $b<0$, where equation (\ref{minus_b_2}) applies, unless there is an 
injection of cash to move to a state with $b>0$.
The terminal condition for equation (\ref{minus_b_2}) is
\begin{eqnarray}
\hat{V}(\hat{b}, W^*,T^+) & = &  \kappa \biggl( 
                               W^* + \frac{\min( (-\hat{b} -W^*), 0 )}{\alpha} 
                               \biggr) + \epsilon (-\hat{b})  ~;~ \hat{b} > 0~. \label{minus_b_3} 
\end{eqnarray}

A brief overview of the numerical algorithms is given in Appendix \ref{Numerical_Appendix}, along
with a numerical convergence verification.

\section{Data}\label{data_section}

We use data from the Center for
Research in Security Prices (CRSP) on a monthly basis over the
1926:1-2020:12 period.\footnote{More specifically, results presented here
were calculated based on data from Historical Indexes, \copyright
2020 Center for Research in Security Prices (CRSP), The University of
Chicago Booth School of Business. Wharton Research Data Services (WRDS) was
used in preparing this article. This service and the data available
thereon constitute valuable intellectual property and trade secrets
of WRDS and/or its third-party suppliers.} Our base case tests use
the CRSP US 30 day T-bill for the bond asset
and the CRSP value-weighted total return index for the stock asset.
This latter index includes all distributions for all domestic stocks
trading on major U.S.\ exchanges.
All of these various indexes are in
nominal terms, so we adjust them for inflation by using the U.S.\ CPI
index, also supplied by CRSP. We use real indexes since investors funding
retirement spending should be focused on real (not nominal) wealth goals.

We use the threshold technique \citep{mancini2009,contmancini2011,Dang2015a}
to estimate the parameters for the parametric stochastic process models.
Since the index data is in real terms, all parameters
reflect real returns.
Table \ref{fit_params} shows the results of calibrating
the models to the historical data.
The correlation $\rho_{sb}$
is computed by removing any returns which occur at times corresponding
to jumps in either series, and then using the sample covariance.
Further discussion of the validity of assuming that the stock and
bond jumps are independent is given in \citet{forsyth_2020_a}.

{\small
\begin{table}[hbt!]
\begin{center}
\begin{tabular}{cccccccc} \toprule
 CRSP & $\mu^s$ & $\sigma^s$ & $\lambda^s$ & $u^s$ &
  $\eta_1^s$ & $\eta_2^s$ & $\rho_{sb}$ \\ \midrule
       & 0.08912  & 0.1460&   0.3263  &  0.2258 & 4.3625 & 5.5335 & 0.08420\\
 \midrule
 \midrule
30-day T-bill & $\mu^b$ & $\sigma^b$ & $\lambda^b$ & $u^b$ &
  $\eta_1^b$ & $\eta_2^b$ & $\rho_{sb}$ \\ \midrule
        & 0.0046 & 0.0130 & 0.5053  &   0.3958 &  65.801 & 57.793 & 0.08420
\\
\bottomrule
\end{tabular}
\caption{Estimated annualized parameters for double exponential jump
diffusion model.  Value-weighted CRSP index, 30-day T-bill index
deflated by the CPI.  Sample period 1926:1 to 2020:12.
}
\label{fit_params}
\end{center}
\end{table}
}

\begin{remark}[Choice of 30-day T-bill for the bond index]
It might be argued that the bond index should hold longer dated bonds, e.g. ten-year
treasuries, which would allow the investor to harvest the term premium.  Long term bonds
have enjoyed high real returns over the last thirty years, due to decreasing real interest
rates during that time period.  However, it is unlikely that this will continue to be
true in the next thirty years.  \citet{Hatch_1985} study the real returns of equities,
short term T-bills, and long term corporate and government bonds, over the period 1950-1983,
and conclude that, in both Canada and the US, only equities and short term T-bills had non-negative
real returns.  Inflation (both US and Canada) averaged about 4.75\% per year
over the period 1950-1983.  If one imagines that the next thirty years will be a period with inflationary
pressures, then this suggests that the defensive asset should be short term T-bills.
However, there is nothing in the methodology in this
paper which prevents us from using other underlying bonds in the bond index.
\end{remark}

\section{Historical Market}\label{boot_section}
We compute and store the optimal controls based on the parametric
model (\ref{jump_process_stock}-\ref{jump_process_bond}) as for the
synthetic market case.  However, we compute statistical quantities
using the stored controls, but using bootstrapped historical return
data directly.  We remind the reader that all returns are
inflation adjusted.  We use the stationary block bootstrap method 
\citep{politis1994,politis2004,politis2009,
Cogneau2010,dichtl2016,Scott_2022,Simonian_2022,Cederburg_2022}.
A key parameter is the expected blocksize. Sampling the data in blocks
accounts for serial correlation in the data series.    We use the algorithm in \citet{politis2009} to determine
the optimal blocksize for the bond and stock returns separately, see Table \ref{auto_blocksize}.
We use a paired sampling approach to
simultaneously draw returns from both time series.  In this case,  a reasonable
estimate for the blocksize for the paired resampling algorithm would
be about $2.0$ years.
We will give results for a range of blocksizes as a check on the robustness of the bootstrap
results.
Detailed pseudo-code for block bootstrap resampling is given in \citet{Forsyth_Vetzal_2019a}.

\begin{table}[tb]
\begin{center}
{\small
\begin{tabular}{lc} \toprule
Data series          & Optimal expected \\
                     & block size $\hat{b}$ (months) \\ \midrule
 Real 30-day T-bill  index    &  50.6 \\
  Real CRSP value-weighted index &  3.42 \\
\bottomrule
\end{tabular}
}
\end{center}
\caption{Optimal expected blocksize $\hat{b}=1/v$ when the blocksize follows
a geometric distribution $Pr(b = k) = (1-v)^{k-1} v$. The algorithm in
\citet{politis2009} is used to determine $\hat{b}$.
Historical data range 1926:1-2020:12.
\label{auto_blocksize}
}
\end{table}

\section{Investment Scenario}
Table \ref{base_case_1} shows our base case investment scenario.
We will use thousands of dollars as our units of wealth in the following.  For example,
a withdrawal of $40$ per year corresponds to $\$40,000$ per
year (all values are real, i.e. inflation
adjusted), with an initial wealth of $1000$ ($\$1,000,000$).    This 
would correspond to the use of the four per cent rule \citep{Bengen1994}.
Our base case scenario assumes a fee of 50 bps per year.  We refer to
\citet{chen_2021_d} for a discussion of tontine fees.

As a motivating example, we consider a 65-year old Canadian retiree who has a pre-retirement salary of \$100,000 per year,
with \$1,000,000 in a DC savings account.  Government benefits are assumed to amount to about \$20,000 per year (real).  The
retiree wishes the DC plan to generate at least \$40,000 per year (real), so that the DC plan and government
benefits  replace 60\% of pre-retirement income.
We assume that the retiree owns mortgage-free real estate worth about \$400,000.  In an act of mental accounting, the
retiree plans to use the real estate as a longevity hedge, which could be monetized using a reverse mortgage.  In the
event that the longevity hedge is not needed, the real-estate will be a bequest.  Of course, the retiree would like to
withdraw more than \$40,000 per year from the DC plan, but has no use for withdrawals greater than \$80,000 per year.
We make the further assumption that the real-estate holdings can generate \$200,000 through a reverse mortgage.  Hence,
as a rough rule of thumb, any expected shortfall at $T=30$ years greater than $-\$200,000$ is an acceptable level of risk.

Our view that personal real estate is not fungible with investment assets (unless investment assets are depleted) is
consistent with the behavioral life
cycle approach originally described in \citet{Shefrin-Thaler:1988} and \citet{Thaler:1990}. In
this framework, investors divide their wealth into separate ``mental
accounts'' containing funds intended for different purposes such as
current spending or future need.

We take the view of a 65-year old retiree, who wants to maximize her total withdrawals, and minimize
the risk of running out of savings, assuming that she lives to the age of 95.
We also assume that the retiree has no bequest motive.

Recall that \citet{Bengen1994} attempted to determine a safe real withdrawal rate, and constant
allocation strategy, such that
the probability of running out of cash after 30 years of retirement was small.
In other words, \citet{Bengen1994} maximized total withdrawals over a 30 year period, 
assuming that the retiree survived for the entire 30 years.  This is, of course, a
conservative assumption.

In our case, we are essentially answering the same question.  The key difference here
is that we (i) allow for dynamic asset allocation, (ii) allow variable withdrawals (within limits)
and (iii) assume a tontine overlay.

\begin{table}[hbt!]
\begin{center}
\begin{tabular}{lc} \toprule
Retiree & 65-year old Canadian male\\
Tontine Gain $\mathbb{T}^g$ & equation (\ref{tontine_gain}) \\
Group Gain $G$ ( see equation (\ref{tontine_gain_rate_special_plus_G} ) ) & 1.0\\
Mortality table & CPM 2014\\
Investment horizon $T$ (years) & 30.0  \\
Equity market index & CRSP Cap-weighted index (real) \\
Bond index & 30-day T-bill (US) (real) \\
Initial portfolio value $W_0$  & 1000 \\
Cash withdrawal$/$rebalancing times & $t=0,1.0, 2.0,\ldots, 29.0$\\
Maximum withdrawal (per year)   & $ q_{\max} = 80$\\
Minimum withdrawal (per year)   & $ q_{\min} = 40$\\
Equity fraction range & $[0,1]$\\
Borrowing spread $\mu_c^{b}$ & 0.02 \\
Rebalancing interval (years) & 1.0  \\
$\alpha$ (EW-ES)               & .05 \\
Fees $\mathbb{T}^f$ ( see equation (\ref{W_plus_def}) )        & 50 bps per year\\
Stabilization $\epsilon$ ( see equation (\ref{PCES_a}) )    &  $ -10^{-4} $ \\
Market parameters & See Table~\ref{fit_params} \\ \bottomrule
\end{tabular}
\caption{Input data for examples.  Monetary units: thousands of dollars.
CPM2014 is the mortality table from the Canadian Institute of Actuaries.
\label{base_case_1}}
\end{center}
\end{table}

\section{Constant Withdrawal, Constant Equity Fraction.}\label{constant_p_q_section}
As a preliminary example, in this section we present results for the scenario in
Table \ref{base_case_1}, except that a constant withdrawal of $40$ per year is specified,
along with a constant weight in stocks at each rebalancing date.  

Table \ref{const_q_p_synthetic_table} gives the results for various values of the
constant weight equity fraction, in the synthetic market.  The best result\footnote{Recall that
$ES$ is defined in terms of the left tail mean of final wealth (not losses) hence a larger
value is preferred.} for
$ES$ (the largest value) occurs at the rather low constant equity weight of $p=0.1$,
with an $ES$ $=-239$.
Table \ref{constant_p_q_historical_table} gives similar results, this time using
bootstrap resampling of the historical data (the historical market).
This time, the best value of $ES=-305$ occurs for a constant equity fraction of $p=0.4$.
Consequently, in both the historical and synthetic market, the constant weight, constant
withdrawal strategy fails to meet our risk criteria of $ES > -200$.

These simulations indicate that there is significant depletion risk for the 
constant withdrawal, constant weight strategy suggested in \citet{Bengen1994}.
{\small
\begin{table}[hbt!]
\begin{center}
{\small
\begin{tabular}{ccccc} \hline
Equity fraction $p$  &  $E[ \sum_i \qq_i]/T$ & ES (5\%)  & $Median[W_T]$   \\ \hline
    0.0          &40         & -302.57      &-150.56 \\
         0.1&     40 &         -238.62 &    -6.82\\
        0.2 &    40  &        -245.48  &      168.10\\
        0.3 &    40  &        -280.27  &     386.05\\
        0.4 &    40  &        -330.37  &    649.58\\
         0.5&    40  &        -391.61  &    958.33\\
       0.6  &    40  &        -461.54  &      1312.17\\
       0.7  &    40  &        -538.04  &     1706.49\\
       0.8  &    40  &        -619.31  &     2135.24\\
\hline
\end{tabular}
}
\caption{Constant weight, constant withdrawals, synthetic market results. No tontine gains.
 Stock index: real capitalization weighted CRSP stocks;
bond index: real 30-day T-bills.  Parameters from Table \ref{fit_params}.
Scenario in Table \ref{base_case_1}.
Units: thousands of dollars. Statistics
based on $2.56 \times 10^6$ Monte Carlo simulation runs.
$T = 30 $ years.
\label{const_q_p_synthetic_table}
}
\end{center}
\end{table}
}

{\small
\begin{table}[hbt!]
\begin{center}
{\small
\begin{tabular}{cccc} \hline
Equity fraction $p$     & $E[ \sum_i \qq_i]/T$  &  ES (5\%)    &  $Median[W_T]$   \\ \hline
          0.0  &          40   &        -508.44 &    -155.04\\
         0.1   &      40       &    -418.02     &  -10.98\\
         0.2   &     40        &    -350.00     &    164.75\\
         0.3   &     40        &     -312.24    &    382.16\\
         0.4   &     40        &    -305.52     &    649.04\\
          0.5  &     40        &      -326.40   &    966.61\\
          0.6  &     40        &     -370.18    &   1336.31\\
          0.7  &     40        &     -432.55    &   1759.66\\
          0.8  &     40        &     -509.00    &   2232.29\\
\hline
\end{tabular}
}
\caption{Constant weight, constant withdrawals, historical market.
No tontine gains.
Historical data range 1926:1-2020:12.
Constant withdrawals are 40 per year.
 Stock index: real capitalization weighted CRSP stocks;
bond index: real 30-day T-bills.
Scenario in Table \ref{base_case_1}.
Units: thousands of dollars. Statistics
based on $10^6$ bootstrap simulation runs.  Expected blocksize $=2$ years.
$T = 30 $ years
\label{constant_p_q_historical_table}
}
\end{center}
\end{table}
}

\section{Synthetic Market Efficient Frontiers}
Figure \ref{synthetic_frontiers_figs_all} shows the efficient EW-ES frontiers,
computed in the synthetic market,
for the following cases:
\begin{description}
   \item[Tontine:] the case in Table \ref{base_case_1}.  The control is computed using the
     algorithm in Section \ref{DP_program} and then stored, and used in Monte Carlo simulations.  The detailed
     frontier is given in Table \ref{synthetic_EW_ES_table_synthetic}. 
   \item[No Tontine:] the case in Table \ref{base_case_1}, but without any tontine gains.
             The control is computed and stored, and then used in Monte Carlo simulations.
             The detailed frontier is given in Table \ref{synthetic_EW_ES_table_synthetic_no_tontine_gains}.
   \item[Const q=40, Const p:] The best single point from Table \ref{const_q_p_synthetic_table}, based on
            Monte Carlo simulations.
\end{description}

Note that all these strategies produce a minimum withdrawal of $40$ per year (i.e. 4\% real of the initial
investment) for thirty years.  However, the best result for the constant weight strategies was $(EW,ES) = (40, -239)$
This can be improved significantly by optimizing over withdrawals and asset allocation, but with no tontine gains.
For example, from Table \ref{synthetic_EW_ES_table_synthetic_no_tontine_gains}, the nearest point with roughly
the same level of risk is $(EW,ES) = (58,-242)$.  However, the improvement with optimal controls and tontine gains
is remarkable.  For example, it seems reasonable to target a value of $ES \simeq 0$.
From Table \ref{synthetic_EW_ES_table_synthetic}, we note the point $(EW,ES) = (69,47)$,
which is dramatically better than any No Tontine Pareto point.
This can also be seen from the large outperformance in the EW-ES frontier compared to the No Tontine case
in Figure \ref{synthetic_frontiers_figs_all} .

%
%
\begin{figure}[htb!]
\centerline{%
\includegraphics[width=3.0in]{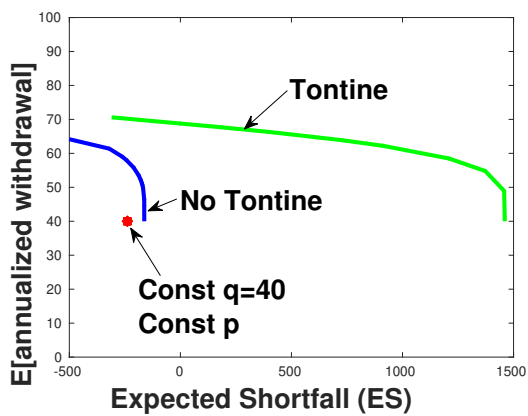}
}
\caption
{Frontiers generated from  the synthetic market.
Parameters based on real CRSP index,
real 30-day T-bills (see Table \ref{fit_params}).
Tontine case is as in Table \ref{base_case_1}.
The No Tontine case uses the same scenario, but with no tontine
gains, and no fees.  
The Const q, Const p case has $q=40$, $p=0.10$, with no
tontine gains, which
is the best result from Table \ref{const_q_p_synthetic_table},
assuming no tontine gains, and no fees.
Units: thousands of dollars.}
\label{synthetic_frontiers_figs_all}
\end{figure}

\subsection{Effect of Fees}
Figure \ref{synthetic_fees_fig} shows the effect of varying the annual fee in the synthetic
market, for the scenario in Table \ref{base_case_1}.  Recall that the base case specified
a fee of 50 bps per year.  Assuming a shortfall target of ES $\simeq 0$, then the effect of
fees in the range $0-100$ bps is quite modest.  Even with annual fees of 100 bps, the Tontine case
still significantly outperforms the No Tontine case (which is assumed to have no fees).

\begin{figure}[htb!]
\centerline{%
\includegraphics[width=3.0in]{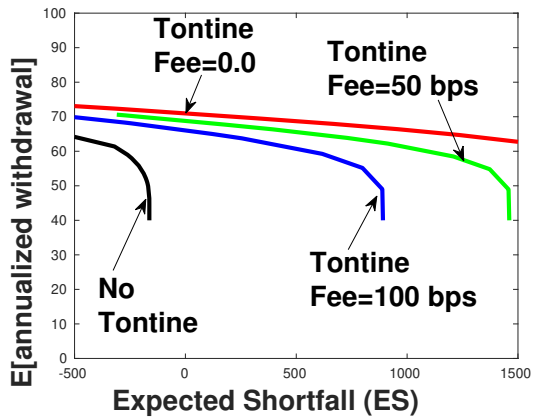}
}
\caption
{Effect of varying fees charged for the Tontine,
basis points (bps) per year.  Frontiers generated from  the synthetic market.
Parameters based on real CRSP index,
real 30-day T-bills (see Table \ref{fit_params}).
Base case Tontine  is as in Table \ref{base_case_1} (fees $50$ bps per year).
The No Tontine case uses the same scenario, but with no tontine
gains, and no fees.
Units: thousands of dollars.}
\label{synthetic_fees_fig}
\end{figure}

\subsection{Effect of Random G}
Recall the definition of the group gain $G_i$ at time $t_i$
in equation (\ref{group_gain}).  Basically, the group gain is used
to ensure that the total amount of mortality credits disbursed is exactly
equal to the total amount forfeited by tontine participants who have
died in $(t_{i-1}, t_i)$.

If Condition \ref{small_bias_assumption} holds, then we expect that the effect of
randomly varying $G_i$ to have a small cumulative effect.  In \citet{Fullmer_2019} and \citet{Winter_2020},
the authors create synthetic tontine pools, where the investors have different initial
wealth, ages, genders, and investment strategies.  These pools are perpetual,
i.e. new members join as original members die.  It is assumed that the investors can
only select an asset allocation strategy from a stock index and a bond
index, both of which follow a geometric Brownian motion (GBM).

The payout rules are different from those
suggested in this paper, however, it is instructive to observe the following.
In \citet{Fullmer_2019}, the perpetual tontine pool has 15,000 investors at steady-state.  After the initial
start-up period, the expected value of the group gain $G_i$ at each rebalancing time
is close to unity, with a standard deviation of about $0.1$.   \citet{Fullmer_2019} also
note that there is essentially no correlation between investment returns, and the
group gain.

Figure \ref{synthetic_random_G_fig} shows the effect of randomly varying $G_i$.
The curve labeled $G=1.0$ is the base case EW-ES curve from the scenario
in Table \ref{base_case_1}, in the synthetic market (parameters in
Table \ref{fit_params}).   The controls from this base case are 
stored, and then used in Monte Carlo simulations, where $G$ is 
assumed to have a normal distribution with mean one, and standard
deviation of $0.1$.  The EW-ES curves for both cases essentially
overlap, except for very large values of $ES$, which are not of any
practical interest.  We get essentially the same result if we
use a uniform distribution for $G$ with $E[G] = 1$, with the same standard deviation.
This is not surprising, since, assuming that the value function is smooth,
then a simple Taylor series argument shows that, for any assumed distribution 
of $G$ with mean one, the effect of randomness of $G$ is a second order effect in
the standard deviation.

Of course, we cannot determine the actual distribution of $G$ without
a detailed knowledge of the characteristics of the tontine pool.
In fact, if we knew the distribution, we could include it in the formulation
of the optimal control problem.  However, knowledge of the distribution
of $G$ 
is unlikely to be available to pool participants in practice.

Nevertheless, the simulations in \citep{Fullmer_2019,Winter_2020}, coupled with our
results as shown in Figure \ref{synthetic_random_G_fig}, suggest that, for a sufficiently
large, diversified pool of investors that the  effects of randomly varying $G$ are negligible.

\begin{figure}[htb!]
\centerline{%
\includegraphics[width=3.0in]{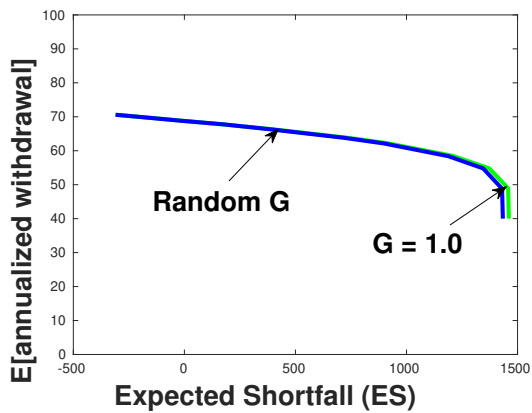}
}
\caption
{Effect of randomly varying group gain $G$ (Section \ref{group_section}).
Frontiers generated from  the synthetic market.
Parameters based on real CRSP index,
real 30-day T-bills (see Table \ref{fit_params}).
Base case Tontine ($G=1.0)$  is as in Table \ref{base_case_1}.
Random G case uses the control computed for the base case,
but in the Monte Carlo simulation, $G$ is normally distributed
with mean one and standard deviation $0.1$.
Units: thousands of dollars.}
\label{synthetic_random_G_fig}
\end{figure}

%
%

\section{Bootstrapped Results}
As discussed in Section \ref{boot_section}, a key parameter in the stationary block bootstrap technique
is the expected blocksize.  In Figure \ref{boot_frontier_fig}, we show the results of the following test. 
We compute and store the optimal controls, based on the synthetic market.  Then we use these controls, but
carry out tests on bootstrapped historical data.  The efficient frontiers in Figure \ref{boot_frontier_fig},
for ES $<1000$ essentially overlap for all expected blocksizes in the range $0.5-5.0$.  
Since it is probably not of interest to aim for an ES of 1000 (which is one million dollars) at age 95,
this indicates that the
computed strategy is robust to parameter uncertainty.  

Figure \ref{boot_alpha_fig} compares the efficient frontier
tested in the historical market (expected blocksize 2 years), with the efficient frontier in the synthetic market.
We observe that the synthetic and historical curves  overlap for $ES < 1000$, which
again verifies that the controls are robust to data uncertainty.
The efficient frontiers/points 
for the No Tontine case and the constant weight, constant withdrawal strategy
(computed in the historical market) are also shown.  The Tontine overlay continues to outperform the No Tontine
case by a wide margin.

\begin{figure}[tb]
\centerline{%
\begin{subfigure}[t]{.40\linewidth}
\centering
\includegraphics[width=\linewidth]{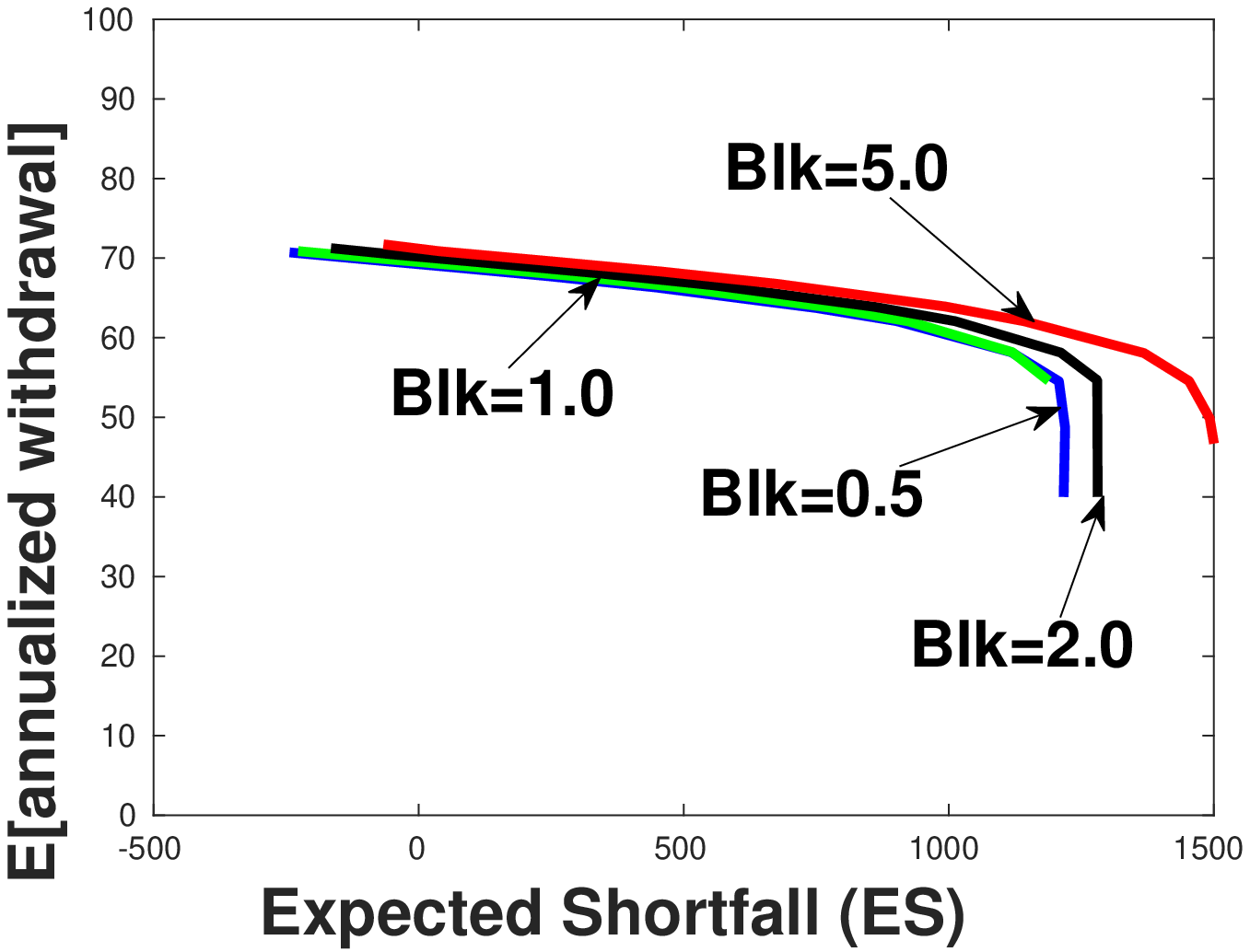}
\caption{Efficient frontiers, historical market, effect of varying expected blocksize (Blk).}
\label{boot_frontier_fig}
\end{subfigure}
\begin{subfigure}[t]{.40\linewidth}
\centering
\includegraphics[width=\linewidth]{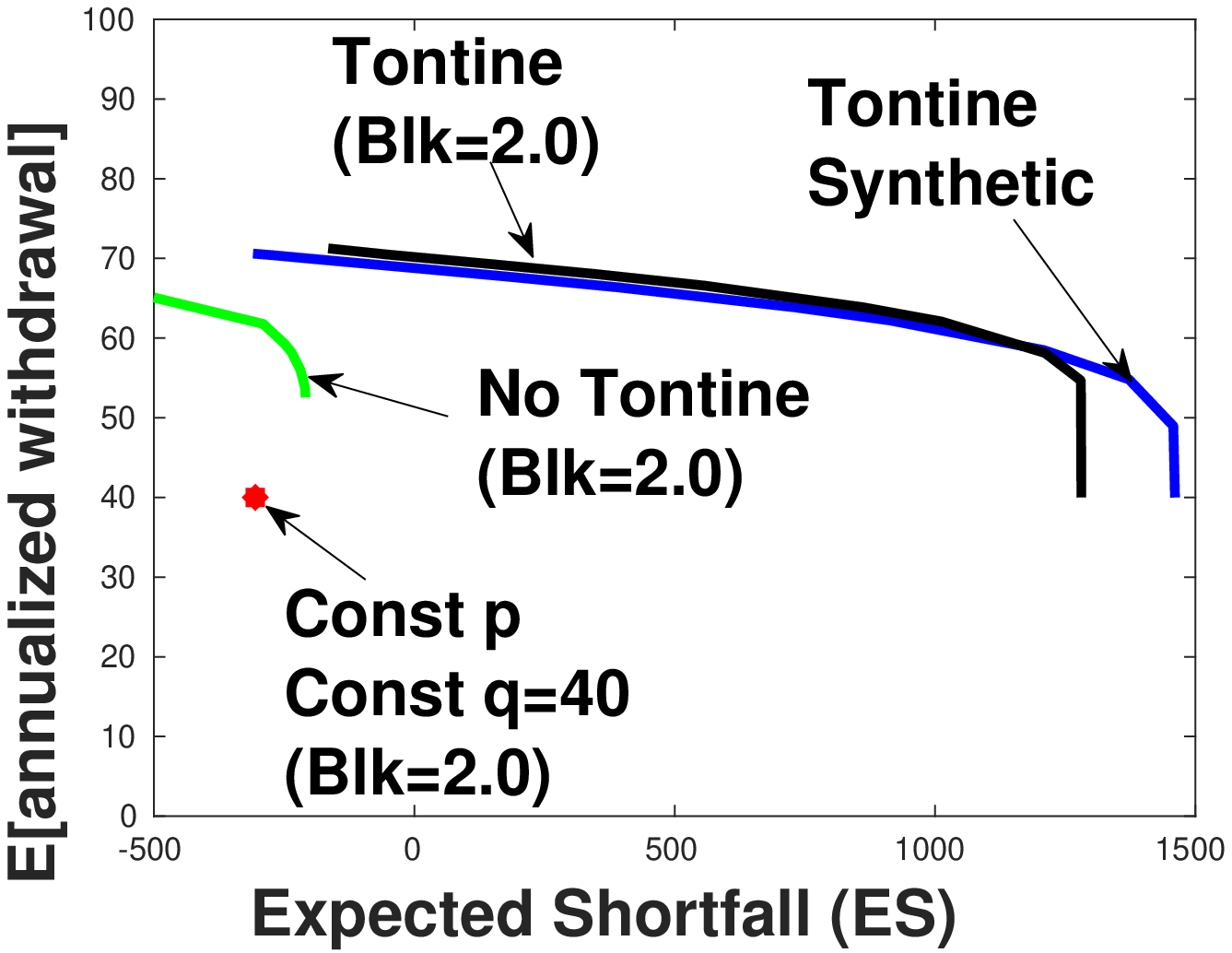}
\caption{Efficient frontiers, historical market, expected blocksize $=2.0$ years. 
Synthetic market frontier (tontine) also shown.
}
\label{boot_alpha_fig}
\end{subfigure}
}
\caption{
Optimal strategy
determined by solving Problem \ref{PCES_a} in the synthetic market,
parameters in Table \ref{fit_params}.  Control stored and then tested
in bootstrapped historical market.
Inflation adjusted data, 1926:1-2020:12.
Non-Pareto points eliminated.  Expected blocksize (Blk, years) used in the bootstrap
resampling method also shown.
Units: thousands of dollars.
The const q, const p case had $(p,q) = (0.4, 40)$ (no tontine gains).  This is the best result for the
constant $(p,q)$ case, shown in Table \ref{constant_p_q_historical_table}.
}

\label{boot_figs}
\end{figure}

\section{Detailed Historical Market Results: EW-ES Controls}
In this section, we examine some detailed characteristics of the optimal
EW-ES strategy, tested in the historical market for the scenario
in Table \ref{base_case_1}.
Figure \ref{percentiles_EW_ES_figs} shows the percentiles of
fraction in stocks, wealth, and withdrawals versus time,
for the EW-ES control with $\kappa = 0.18$,
with $(EW,ES) = (69, 204)$.
To put this in perspective, recall that this strategy never withdraws
less than $40$ per year.  Compare this to the best case for a constant
withdrawal, constant weight strategy (no tontine) from Table \ref{constant_p_q_historical_table},
which has $(EW,ES) = (40, -306) $,
or to the optimal EW-ES strategy, but with no tontine, from Table \ref{historical_EW_ES_table_no_tontine},
which has $(EW,ES) = (70,-806)$.

Figure \ref{percentile_control_fig} shows that the median optimal fraction in stocks starts at about $0.60$, then drops to
about $0.20$ at 15 years, finally ending up at zero in year 26.  Figure \ref{percentile_wealth_fig} indicates that for
the years in the span of $20-30$, the median and fifth percentiles of wealth are fairly tightly clustered, with
the fifth percentile being well above zero at all times.  The optimal withdrawal percentiles  are shown
in Figure \ref{percentile_qplus_fig}.  The median withdrawal starts at $40$ per year,  then increases
to the maximum withdrawal of $80$ in years $3-4$, and remains at $80$ for the remainder of the thirty year
time horizon.

\begin{figure}[tb]
\centerline{%
\begin{subfigure}[t]{.33\linewidth}
\centering
\includegraphics[width=\linewidth]{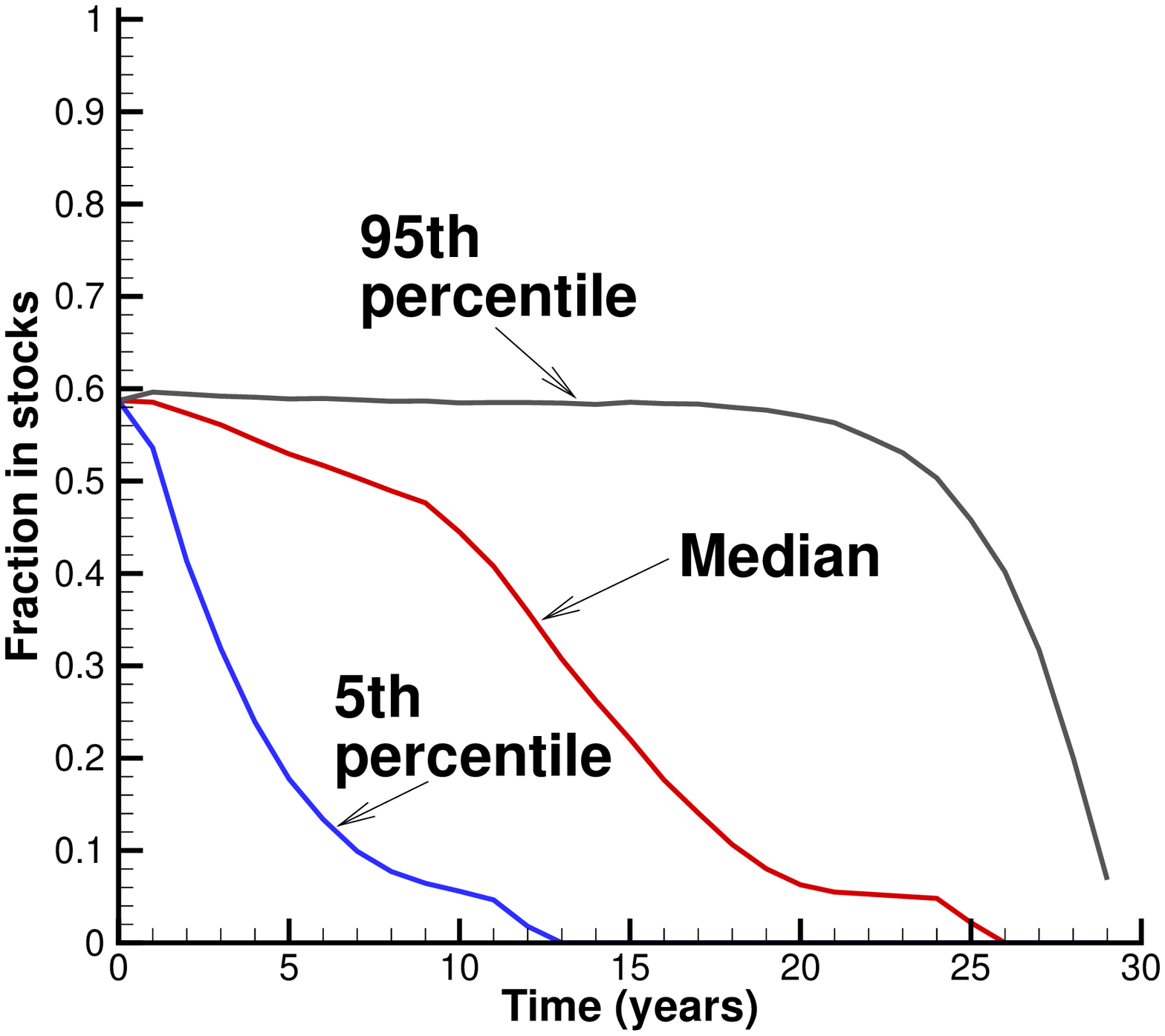}
\caption{Percentiles fraction in stocks}
\label{percentile_control_fig}
\end{subfigure}
\begin{subfigure}[t]{.33\linewidth}
\centering
\includegraphics[width=\linewidth]{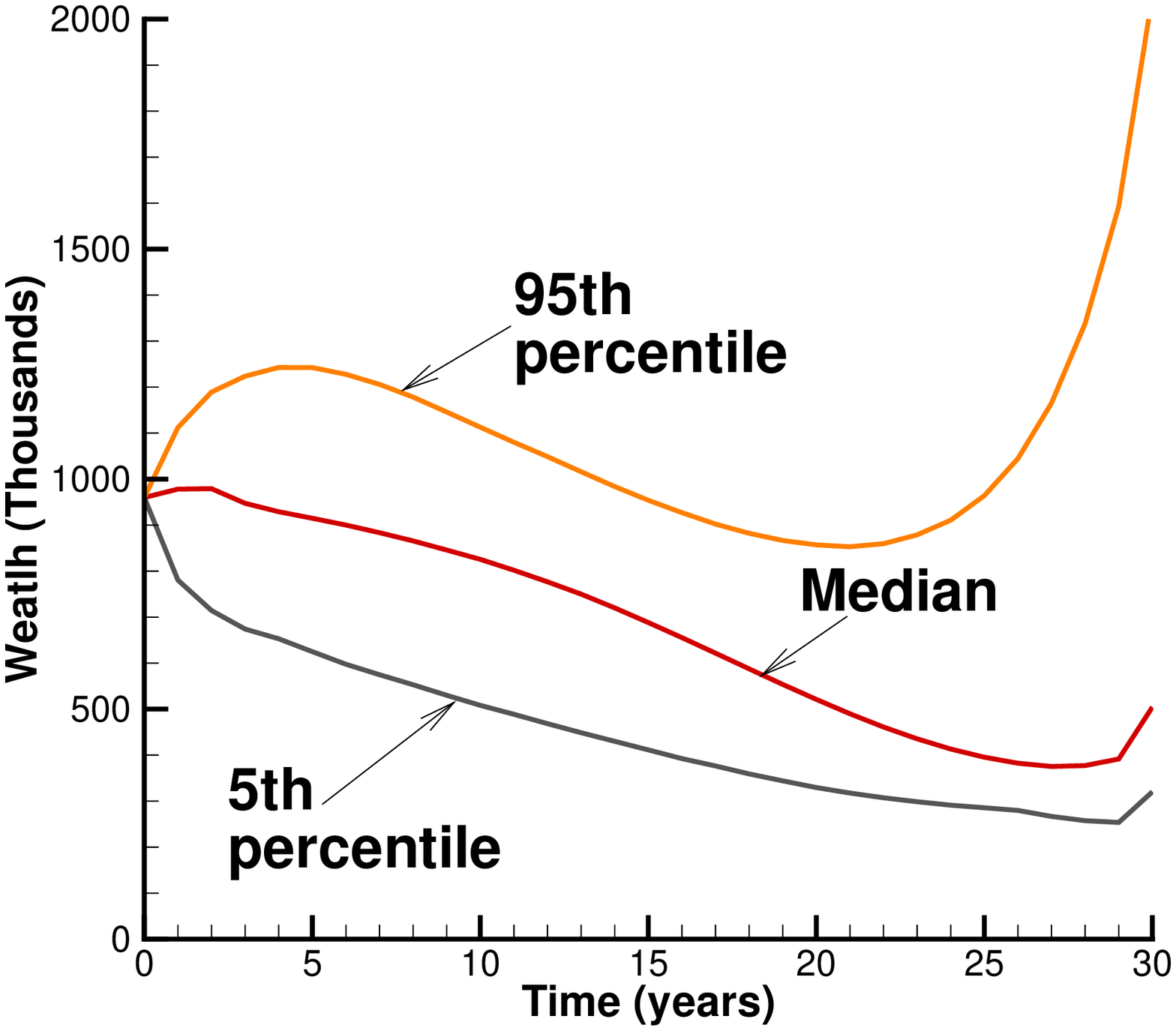}
\caption{Percentiles  wealth}
\label{percentile_wealth_fig}
\end{subfigure}
\begin{subfigure}[t]{.33\linewidth}
\centering
\includegraphics[width=\linewidth]{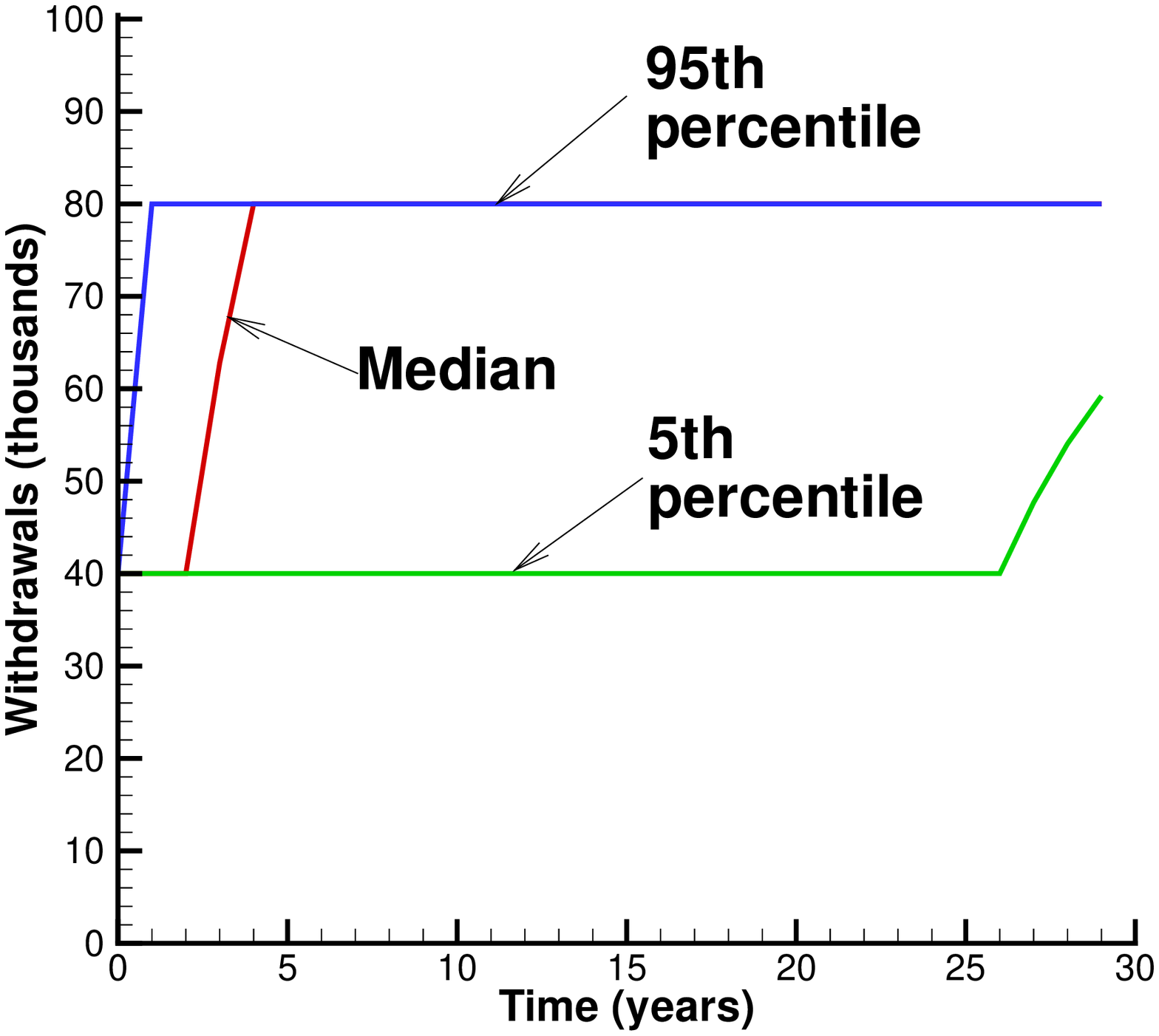}
\caption{Percentiles withdrawals}
\label{percentile_qplus_fig}
\end{subfigure}
}
\caption{Scenario in Table  \ref{base_case_1}.
EW-ES control
computed from problem EW-ES Problem (\ref{PCES_a}).
Parameters based on the real CRSP index,
and real 30-day T-bills (see Table \ref{fit_params}).  Control computed and stored
from the Problem (\ref{PCES_a}) in the synthetic market.
Control used in the historical market, $10^6$ bootstrap samples.
$q_{min} = 40, q_{\max} = 80$ (per year), ${\kappa}= 0.18$. $W^* = 385$.
Units: thousands of dollars.
}
\label{percentiles_EW_ES_figs}
\end{figure}

Figure \ref{heat_map_EW_LSfig_all} shows the optimal control
heat maps for the fraction in stocks and withdrawal amounts, for the scenario
in Table \ref{base_case_1}.  Figure \ref{heat_map_EW_ES_fig} shows a smooth behavior
of the optimal fraction in stocks as a function of $(W,t)$.  This can
be compared with the equivalent heat map for the EW-ES control 
in \citet{Forsyth_2021_b} (no tontine gains), which is much more aggressive at changing
the asset allocation in response to changing wealth amounts.  The smoothness
of the controls in Figure \ref{heat_map_EW_ES_fig} appears to be  due
to the rapid de-risking of a strategy which includes tontine gains, which
provides a natural protection against sudden stock index drops.
The upper blue zone in Figure \ref{heat_map_EW_ES_fig} is de-risking due to
the fact that, with sufficiently large wealth, there is essentially no probability
of running out of cash even at the maximum withdrawal amount.
The use of the stabilization factor $\epsilon = -10^{-4}$ forces the strategy
to increase the weight in bonds for large values of wealth (see equation (\ref{PCES_a})).\footnote{``When you have won
the game, stop playing,'' William Bernstein.}
The lower red zone is in response to extremely poor wealth outcomes, which
means that the optimal strategy 
is to invest 100\% in stocks and hope for the best.  However, this
is an extremely unlikely outcome, as can be verified from
Figure \ref{percentile_wealth_fig}.

From Figure \ref{heat_map_qplus_EW_ES_fig}, we can observe that the optimal withdrawal strategy
is essentially a bang-bang control, i.e. withdraw at either the maximum or minimum amount per
year.  This is not unexpected, as discussed in Appendix
\ref{appendix_limit}.  We also note that this
type of strategy  has been suggested previously, based on heuristic reasoning.\footnote{
      {\em{``If we have a good year,
         we take a trip to China,...if we have a bad year, we stay home and play canasta,''}}
         retired professor Peter Ponzo, discussing his DC plan withdrawal strategy
    {\url{https://www.theglobeandmail.com/report-on-business/math-prof-tests-investing-formulas-strategies/article22397218/}}
       .}

\begin{figure}[htb!]
\centerline{
\begin{subfigure}[t]{.4\linewidth}
\centering
\includegraphics[width=\linewidth]{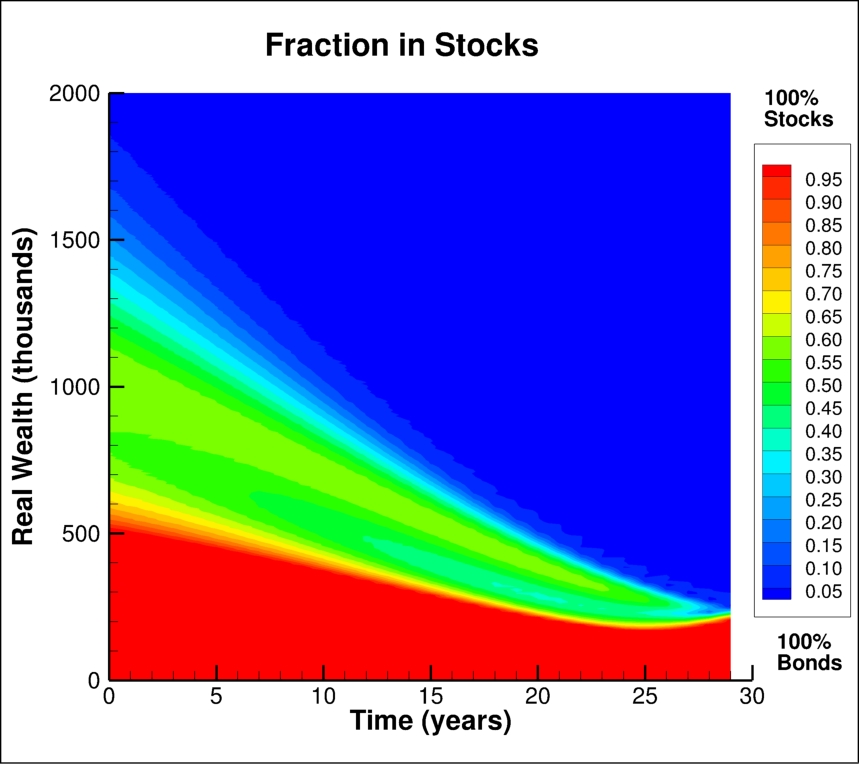}
\caption{Fraction in stocks}
\label{heat_map_EW_ES_fig}
\end{subfigure}
\hspace{.05\linewidth}
\begin{subfigure}[t]{.4\linewidth}
\centering
\includegraphics[width=\linewidth]{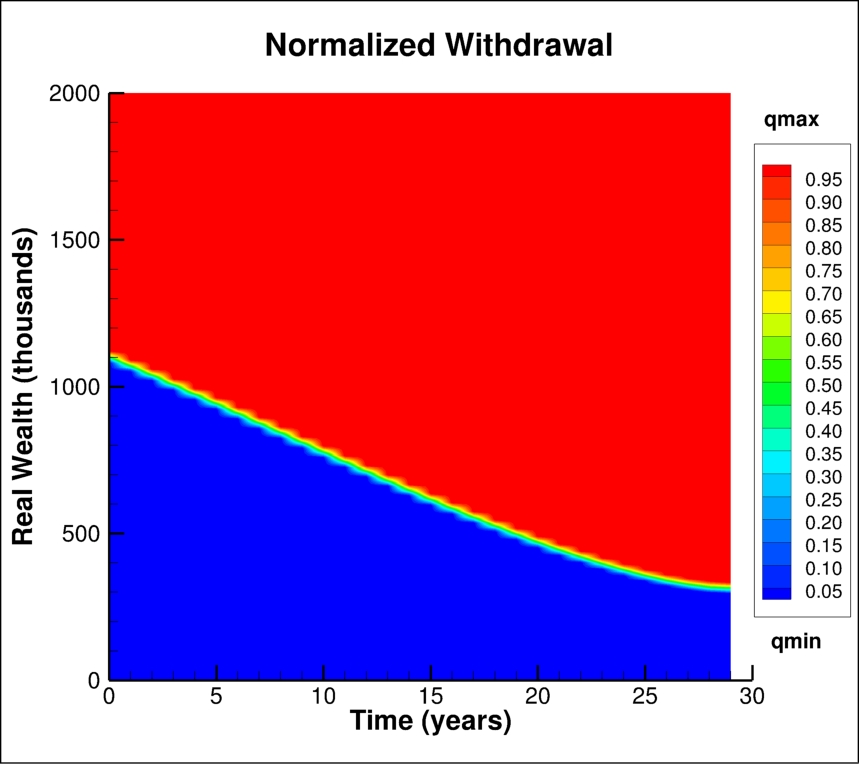}
\caption{Withdrawals}
\label{heat_map_qplus_EW_ES_fig}
\end{subfigure}
}
\caption{
Optimal EW-ES.
Heat map of controls: fraction in stocks and
withdrawals, computed from   Problem EW-ES (\ref{objective_overview}).
Real capitalization weighted CRSP index, and real 30-day T-bills.
Scenario given in Table~\ref{base_case_1}.
Control computed and stored from the Problem \ref{PCES_a} in the synthetic market.
$q_{min} = 40, q_{\max} = 80$ (per year). ${{\kappa}} = 0.18$. $W^* = 385$.
$\epsilon = -10^{-4}$.
Normalized withdrawal $(q - q_{\min})/(q_{\max} - q_{\min})$.
Units: thousands of dollars.
}
\label{heat_map_EW_LSfig_all}
\end{figure}

\section{Discussion}
Traditional annuities with true inflation protection are unavailable 
in Canada.\footnote{Some providers advertise annuities with inflation protection,
however this is simply an escalating nominal payout, based on a fixed escalation rate.}
Since inflation is expected to be a major factor in the coming years,
inflation protection is a valuable aspect of the optimal EW-ES strategy,
with a tontine overlay.\footnote{Examination of historical periods of high
inflation suggests that a portfolio of short term T-bills and an equal
weight stock index generates significant positive real returns,
see \citet{forsyth_white_paper_2022_inflation}.}
This strategy has an expected {\em real} withdrawal rate, over thirty
years, of about 7\% of the initial capital (per annum), never withdraws
less than 4\% of initial capital per annum, and a positive ES (expected
shortfall) at the 5\% level after thirty years.

Consequently, if we consider a retiree with no bequest motive, then joining a tontine
pool, and following an optimal EW-ES strategy, is certainly an excellent alternative
to  a life annuity. 
Hence, it could be argued that, going forward, the EW-ES optimal  tontine pool strategy has
less risk than a conventional annuity.

Of course, there is no free lunch here.  The reason that  the tontine approach
has a higher mean (and median) payout is that it is not guaranteed.  There is some
flexibility in the withdrawal amounts, and the portfolio contains risky assets.
However, the ultimate risk, as measured by the expected shortfall at year thirty,
is very small.  We can also see that the median payout rises rapidly to the
 maximum withdrawal rate (8\% real of the initial investment) within 3-4 years
of retirement, and stays at
the maximum for the remainder of the thirty year horizon.

As well, the investor forfeits the entire portfolio in the event of
death.  Although this is often considered a drawback, we remind the reader
that annuities and defined benefit (DB) plans have this same property (restricting
attention to a single retiree with no guarantee period).\footnote{Moshe Milevsky, an advocate of modern tontines,
is quoted in the Toronto Star (April 13, 2021) as noting that
{\em``If you give up some of your money when you  die,
you can get more when you are alive.''} }   Of course, it is possible
to overlay various guarantees on to the tontine pool, e.g. a guarantee period, a money back
guarantee, or joint and survivor benefits.  The cost of these guarantees would, 
of course,  reduce the expected annual
withdrawals.

These results are robust to fees in the range of 50-100 bps per year. The long term
results are also insensitive to random group gains.\footnote{The randomness of the
group gain is due to fact that real tontine pools will be finite and heterogeneous.}

However, the tontine gains (after fees) are comparatively small for retirees in
the 65-70 age range.  This suggests that it may be optimal to delay joining
a tontine until  the investor has attained an age of 70 or more.

Although we have explicitly excluded a bequest motive from our considerations, note that
the median withdrawal strategy rapidly ramps up to the maximum withdrawal within
a few years of retirement, and remains there for the entire remaining retirement
period.  Although it is commonly postulated that retirement consumption follows a
{\em U-shaped} pattern, recent studies indicate that real retirement consumption
falls with age (at least in countries which do not have large end of life expenses)\citep{Brancatti_2015}.
In this case, the withdrawals which occur towards the end of the retirement period may
exceed consumption.  This allows the retiree to use these excess cash flows as a living
bequest to relatives or charities.

\section{Conclusions}
DC plan decumulation strategies are typically based on some variant
of the {\em four per cent rule} \citep{Bengen1994}.  However, bootstrap tests
of these rules using historical data show a significant risk of running
out of savings at the end of a thirty year retirement planning horizon.

This risk can be significantly reduced by using optimal stochastic control
methods, where the controls are the asset allocation strategy and the
withdrawal amounts (subject to maximum and minimum constraints)\citep{Forsyth_2021_b,Forsyth_Arva_a}.

However, if we assume the retiree couples an optimal allocation/withdrawal strategy
with participation in a tontine fund, then the risk of portfolio depletion
after 30 years is virtually eliminated.  At the same time, the cumulative
total withdrawals are considerably increased compared with the previous
two strategies.  Of course, this comes at a price: the retiree forfeits
her portfolio upon death.  Hence the tontine overlay is most appealing
to investors who have no bequest motivation, or who manage bequests using other funds.

We should also note that individual tontine accounts allow for complete flexibility
in asset allocation strategies and do not require
purchase of expensive investment products.  These accounts are essentially peer-to-peer
longevity risk management tools. Consequently, the custodian of these
accounts bears no risk, and incurs only bookkeeping 
costs.   Hence the fees charged by the custodian of these accounts
can be very low.   If desired, the retiree can pay for additional
investment advice in a completely transparent manner.

\section{Acknowledgements}
Forsyth's work was supported by the Natural Sciences and Engineering Research Council of
Canada (NSERC) grant RGPIN-2017-03760.  Vetzal's work was supported by a
Canadian Securities Institute Research Foundation grant 2021-2022.

\section{Declaration}
The authors have no conflicts of interest to report.

\clearpage

\appendix
\section*{Appendix}

\section{Induced Time Consistent Strategy}\label{appendix_induced}

Denote the investor's initial wealth at $t_0$ by $W_0^-$.  Then we have the following
result:
\begin{proposition}[Pre-commitment strategy equivalence to a time consistent
policy for an alternative objective function]\label{equiv_thm}
The pre-commitment EW-ES strategy $\mathcal{P}^*$ determined by solving
$J(0, W_0, t_0^-)$ (with $\mathcal{W}^*( 0, W_0^-)$ from equation (\ref{pcee_argmax}))
is the time consistent strategy for the equivalent
problem $TCEQ$ (with fixed $\mathcal{W}^*(0,W_0^-)$), with value function $\tilde{J}(s,b,t)$
defined by 
\begin{eqnarray}
 \left(\mathit{TCEQ_{t_n}}\left(\kappa / \alpha  \right)\right):
    \qquad \tilde{J}\left(s,b,  t_{n}^-\right)  
    & = &
            \sup_{\mathcal{P}_{n}\in\mathcal{A}}
        \Biggl\{
               E_{\mathcal{P}_{n}}^{X_n^+,t_{n}^+}
           \Biggl[ ~\sum_{i=n}^{M} \mathfrak{q}_i ~  + ~
               \frac{\kappa}{\alpha} \min (W_T -\mathcal{W}^*(0, W_0^-),0) 
                    \Biggr. \Biggr. \nonumber \\
         & &  \Biggl. \Biggl. ~~~~~ 
                \bigg\vert X(t_n^-) = (s,b)
                   ~\Biggr] \Biggr\}~.
             \label{timec_equiv}
\end{eqnarray}

\end{proposition}
\begin{proof}
This follows  similar steps as in \citet{forsyth_2019_c}, proof of Proposition 6.2, with the exception
that the reward in \citet{forsyth_2019_c} is expected terminal wealth, while here the reward is
total withdrawals.
\end{proof}

\begin{remark}[An Implementable Strategy]
Given an initial level of wealth $W_0^-$ at $t_0$, then the optimal control\footnote{To be perfectly 
precise here, in the event that the control is non-unique, we impose a tie-breaking
strategy to generate a unique control.}
for the pre-commitment problem (\ref{PCES_a}) is the same optimal control
for the time consistent problem\footnote{Assuming that the same
tie breaking strategy is used as for the pre-commitment
problem.} $\left(\mathit{TCEQ_{t_n}}\left(\kappa / \alpha  \right)\right)$
(\ref{timec_equiv}), $\forall t >0$.  Hence we can regard problem
$\left(\mathit{TCEQ_{t_n}}\left(\kappa / \alpha  \right)\right)$
as the {\em EW-ES induced time consistent strategy}.
Thus, the induced strategy is implementable,
in the sense that the investor has no incentive to deviate
from the strategy computed at time zero, at later times \citep{forsyth_2019_c}.
\end{remark}

\begin{remark}[EW-ES Induced Time Consistent Strategy]
In the following, we will consider the  actual strategy followed by
the investor for any $t>0$ as given by the induced time consistent
strategy $\left(\mathit{TCEQ_{t_n}}\left(\kappa / \alpha  \right)\right)$
in equation (\ref{timec_equiv}), with a fixed value of $\mathcal{W}^*(0, W_0^-)$,
which is identical to the EW-ES strategy at time zero.
Hence, we will refer to this strategy in the following as the EW-ES strategy, 
with the understanding that this refers to strategy 
$\left(\mathit{TCEQ_{t_n}}\left(\kappa / \alpha  \right)\right)$
for any $t>0$.
\end{remark}

\section{Numerical Techniques}\label{Numerical_Appendix}
We solve problems (\ref{PCES_a})  using the techniques described in
detail in \citet{forsythlabahn2017,forsyth_2019_c,Forsyth_2021_b}.
We give only a brief overview here.

We localize the infinite domain to $(s,b) \in [s_{\min}, s_{\max}] \times [b_{\min}, b_{\max}]$, and
discretize $[b_{\min},b_{\max}]$ using  an equally spaced $\log b$ grid, with $n_b$ nodes.
Similarly, we discretize $[s_{\min}, s_{\max}]$ on an equally spaced $\log s$ grid, with $n_s$ nodes.
Localization errors are minimized using the domain extension method in \citep{forsythlabahn2017}.

At rebalancing dates, we solve the local optimization problem  (\ref{dynamic_c}) by discretizing $(\qq ( \cdot), \pp(\cdot) )$ and
using exhaustive search.
Between rebalancing dates, we solve the two dimensional partial integro-differential equation (PIDE) (\ref{expanded_7} )
using Fourier methods \citep{forsythlabahn2017,Forsyth_2021_b}.
Finally, the optimization problem (\ref{final_step_EWES})
is solved using a one-dimensional optimization technique.

We used the value $\epsilon = -10^{-4}$ in equation (\ref{expanded_1}),
which forces the investment strategy to be bond heavy if the remaining
wealth in the investor's account is large, and $t \rightarrow T$.
Using this small value of  gave the same results as
$\epsilon = 0$ for the summary statistics, to four digits.  This is simply because
the states with very large wealth have low probability.  However, this stabilization
procedure produced  smoother heat maps for large wealth values, without altering
the summary statistics appreciably.

\subsection{Convergence Test: Synthetic Market}
We compute and store the optimal controls from solving Problem \ref{PCES_a}
using the parametric
model of the stock and bond processes.  We then use the stored controls in Monte Carlo
simulations to generate statistical results.  As a robustness check, we also use the stored
controls and simulate results using bootstrap resampling of historical data.

Table \ref{conservative_accuracy} shows a detailed convergence test for the base case
problem given in Table \ref{base_case_1}, for the EW-ES problem.  The results are given for
a sequence of grid sizes, for the dynamic programming algorithm 
in Section \ref{DP_program} and Appendix \ref{Numerical_Appendix}.
The dynamic programming algorithm appears to converge at roughly a second
order rate.  The optimal control computed using dynamic programming is
stored, and then used in Monte Carlo computations.  The MC results are in
good agreement with the dynamic programming solution.

For all the numerical examples, we will use the $2048 \times 2048$  grid, since this
seems to be accurate enough for our purposes.

\begin{table}[hbt!]
\begin{center}
\begin{tabular}{lccc|cc} \toprule
 & \multicolumn{3}{|c|}{Algorithm in Section \ref{DP_program} and Appendix \ref{Numerical_Appendix} } & \multicolumn{2}{c}{Monte Carlo} \\ \midrule
Grid &   ES (5\%) & $E[ \sum_i \qq_i]/M$  & Value Function &   ES (5\%) & $E[ \sum_i \qq_i]/M$  \\
\midrule
$512 \times 512$ &  108.13   & 67.99     & 2059.60     & 123.26    & 68.04        \\
$1024 \times 1024$ &  158.88 & 67.79     & 2063.19     & 164.45    & 67.81       \\ 
$2048 \times 2048$ &  201.88  & 67.56   & 2064.27    & 203.87    & 67.56        \\
$4096 \times 4096$ & 206.56    & 67.54   & 2064.54    &  207.70   &  67.54      \\
\midrule
\bottomrule
\end{tabular}
\caption{Convergence test,
real stock index: deflated real capitalization weighted CRSP, real bond index: deflated
30 day T-bills.  Scenario in Table \ref{base_case_1}.
Parameters in Table \ref{fit_params}.
The Monte Carlo method used
$2.56 \times 10^6$ simulations.  The MC method used the control
from the algorithm in Section \ref{DP_program}.
$\kappa = 0.185, \alpha = .05$
Grid refers to the grid
used in the Algorithm in Section \ref{Numerical_Appendix}: $n_x \times n_b$, where $n_x$ is
the number of nodes in the $\log s$ direction, and
$n_b$ is the number of nodes in the $\log b$ direction.
Units: thousands of dollars (real).
$M$ is the total number of withdrawals (rebalancing dates).
\label{conservative_accuracy}}
\end{center}
\end{table}

\begin{table}[hbt!]
\begin{center}
\begin{tabular}{lccc|cc} \toprule
  & \multicolumn{3}{c|}{Algorithm in Section \ref{DP_program} and Appendix \ref{Numerical_Appendix} } & \multicolumn{2}{c}{Monte Carlo} \\ \midrule
Grid &   ES (5\%) & $E[ \sum_i \qq_i]/T$  & Value Function &   ES (5\%) & $E[ \sum_i \qq_i]/M$  \\
\midrule
$512 \times 512$ & -203.31 &54.08  &   860.033  &    -191.99 & 53.96     \\
$1024 \times 1024$ & -191.40 &53.58 &  889.613  &    -188.07 &53.53     \\
$2048 \times 2048$ &  -188.91 &53.57 & 898.712    &    -188.14 &53.55   \\
$4096 \times 4096$ &  -188.04 & 53.54 &  901.106   &  -187.95 & 53.53   \\
\midrule
\bottomrule
\end{tabular}
\caption{No tontine case. Convergence test,
real stock index: deflated real capitalization weighted CRSP,  real bond index: deflated
30 day T-bills.  Scenario in Table \ref{base_case_1}, but no tontine.
Parameters in Table \ref{fit_params}.
The Monte Carlo method used
$2.56 \times 10^6$ simulations.
The MC method used the control
from the algorithm in Section \ref{DP_program}.
$\kappa = 3.75, \alpha = .05$
Grid refers to the grid
used in the Algorithm in Section \ref{Numerical_Appendix}: $n_x \times n_b$, where $n_x$ is
the number of nodes in the $\log s$ direction, and
$n_b$ is the number of nodes in the $\log b$ direction.
Units: thousands of dollars (real).
$M$ is the total number of withdrawals (rebalancing dates).
$W^* = -106.476$ on the finest grid.
\label{conservative_accuracy_no_tontine}}
\end{center}
\end{table}

\section{Continuous Withdrawal/Rebalancing Limit}\label{appendix_limit}
In order to develop some intuition about the nature of the optimal controls,
we will examine the limit as the rebalancing interval becomes vanishingly small.

\begin{proposition}[Bang-bang withdrawal control in the continuous withdrawal limit]
\label{bang_bang_prop}
Assume that
\begin{itemize}
   \item the stock and bond processes follow (\ref{jump_process_stock})
and (\ref{jump_process_bond}),

  \item the portfolio is continuously rebalanced, and withdrawals occur at
     a continuous (finite) rate $\hat{\qq} \in [\hat{\qq}_{\min}, \hat{\qq}_{\max}]$,

  \item the HJB equation for the EW-ES  problem in the continuous rebalancing limit has
        bounded derivatives w.r.t. total wealth,

  \item in the event of ties for the control $\hat{\qq}$, the smallest withdrawal is selected,
\end{itemize}
then the optimal withdrawal control $\hat{\qq}^*(\cdot)$ for the EW-ES problem 
$(PCES_{t_0}(\kappa))$  and for the EW-LS problem $\left(\mathit{EWLS_{t_0}}\left( \hat{\kappa} \right)\right)$
is bang-bang,
$\hat{\qq}^* \in \{\hat{\qq}_{\min}, \hat{\qq}_{\max} \}$.
\end{proposition}

\begin{proof}
This follows the same steps as in \citet{Forsyth_2021_b}.
\end{proof}

\begin{remark}[Bang-bang control for discrete rebalancing/withdrawals]
\label{quasi_prop}
Proposition \ref{bang_bang_prop} suggests that, for sufficiently 
small rebalancing intervals, we can expect the
optimal $\qq$ control (finite withdrawal amount) to be bang-bang,
i.e. it is only optimal to withdraw  either the maximum
amount ${\qq}_{\max}$ or the minimum amount ${\qq}_{\min}$.  
However, it is not clear that this will continue to
be true for the case of yearly rebalancing (which we specify in our numerical examples),
and finite amount controls $\qq$.  In fact,
we do observe that the finite amount control $\qq$ is very close to bang-bang in our numerical experiments,
even for yearly rebalancing. 
\end{remark}

\section{Detailed Efficient Frontiers: Synthetic Market}
\label{detailed_frontier_appendix_1}

{\small
\begin{table}[hbt!]
\begin{center}
{\small
\begin{tabular}{ccccc} \hline
${\kappa}$  & $E[ \sum_i \qq_i]/T$  & ES (5\%) & $Median[W_T]$ &$ W^*$    \\ \hline
 0.15  &          70.06  &     -309.569  &    189.48    &     .490\\
    0.170     &    70.04 &    -270.13    & 185.19    &     .489\\
    .180      &    68.51 &      46.77    &599.42     &    385.28\\
    .185      &    67.56 &      203.87 &   820.65     &   585.97\\
    .20       &    66.41 &      384.76 &   1058.40    &   802.40\\
    0.25      &    63.85 &      732.34 &   1517.04    &  1220.33\\
    0.30      &    62.22 &      912.29 &   1754.40    &  1439.83\\
    0.50      &    58.48 &      1209.40&    2120.59   &  1802.19\\
    1.0       &    54.81 &      1372.46&   2327.42    &  2021.22\\
    10.0      &    48.96 &      1457.52&   2484.58    &2151.79\\
    $\infty$  &    40    &    1460.76    &2885.85 &      2173.04\\
\hline
\end{tabular}
}
\caption{EW-ES synthetic market results for optimal strategies,  assuming
the scenario given in Table~\ref{base_case_1}. Tontine gains assumed.  Stock index: real capitalization weighted CRSP stocks;
bond index: real 30-day T-bills.  Parameters from Table \ref{fit_params}.
Units: thousands of dollars. Statistics
based on $2.56 \times 10^6$ Monte Carlo simulation runs.
Control is computed using the Algorithm in Section \ref{DP_program} and Section \ref{Numerical_Appendix}, stored, and then used
in the Monte Carlo simulations.
$q_{\min} = 0.40$, $q_{\max} = 80$ (annually).
$T = 30 $ years,
$\epsilon = -10^{-4}$.
\label{synthetic_EW_ES_table_synthetic}
}
\end{center}
\end{table}
}

{\small
\begin{table}[hbt!]
\begin{center}
{\small
\begin{tabular}{ccccc} \hline
${\kappa}$  & $E[ \sum_i \qq_i]/T$  & ES (5\%) & $Median[W_T]$ &$ W^*$    \\ \hline
0.180   &   69.17     &       -823.76  &   -2.51   &   -691.81\\
 1.0    &    61.38    &        -319.66 &    -39.47  &  -229.18\\
 1.5    &   58.98     &        -260.92 &    -65.88  & -179.60\\
 1.75   &   57.97     &        -242.34 &   -74.74   &-161.25\\
 2.5    &   55.86     &        -211.03 &   -81.44 &  -132.87\\
 3.75   &   53.55     &        -188.14 &  -81.11  &-107.00\\
 5.0    &  52.08      &        -177.88 &   -78.39 &  -90.10\\
 6.25   &  51.29      &        -173.59 &    -79.08&   -89.03\\
  7.5   &  50.72      &        -171.05 &     -79.3&     -88.25\\
  10.0  &   49.89     &        -168.16 &     -78.77 &   -87.18\\
   100.0 &   46.41    &         -162.86&       -68.28 &  -77.47\\
   $\infty$ &    40.0 &         -162.67  &     +5.72   &           -76.0\\
\hline
\end{tabular}
}
\caption{EW-ES synthetic market results for optimal strategies,  assuming
the scenario given in Table~\ref{base_case_1}. No tontine gains assumed.  Stock index: real capitalization weighted CRSP stocks;
bond index: real 30-day T-bills.  Parameters from Table \ref{fit_params}.
Units: thousands of dollars. Statistics
based on $2.56 \times 10^6$ Monte Carlo simulation runs.
Control is computed using the Algorithm in Section \ref{DP_program} and Appendix \ref{Numerical_Appendix}, stored, and then used
in the Monte Carlo simulations.
$q_{\min} = 0.40$, $q_{\max} = 80$ (annually).
$T = 30 $ years,
$\epsilon = -10^{-4}$.
\label{synthetic_EW_ES_table_synthetic_no_tontine_gains}
}
\end{center}
\end{table}
}

\clearpage

\section{Detailed Efficient Frontiers: Historical Market}
\label{detailed_appendix_bootstrap}

{\small
\begin{table}[hbt!]
\begin{center}
{\small
\begin{tabular}{ccccc} \hline
${\kappa}$ &  $E[ \sum_i \qq_i]/T$  & ES (5\%) & $Median[W_T]$   \\ \hline
\hline
0.15 &  71.25  &   -165.23    &157.16\\
    .170 &  71.01  &  -138.15 &153.13\\
    .180 &  68.94  &   204.20 & 573.29\\
    .185 &  67.99  &   369.26 & 769.96\\
     .20 &  66.64  &    546.98&  1038.07\\
     .25 &  63.84  &    863.20&   1500.51\\
    0.30 &  62.08  &    1011.55 & 1739.21\\
    0.5  &  58.13  &    1211.18 & 2115.22\\
    1.0  &   54.50 &   1285.93  &  2330.33\\
    10.0 &   49.42 &  1275.98   &2485.58\\
$\infty$ &   40  &    1280.97&   2892.41  \\ 
\end{tabular}
}
\caption{EW-EW historical market results for optimal strategies,  assuming
the scenario given in Table~\ref{base_case_1}. Tontine gains assumed.  Stock index: real capitalization weighted CRSP stocks;
bond index: real 30-day T-bills.  Parameters from Table \ref{fit_params}.
Units: thousands of dollars. Statistics
based on $10^6$  bootstrap simulation runs. Expected blocksize $=2$ years.
Control is computed using the Algorithm in Section \ref{DP_program} and Appendix \ref{Numerical_Appendix}, stored, and then used
in the bootstrap simulations.
$q_{\min} = 40$, $q_{\max} = 80$ (annually).
$T = 30 $ years,
$\epsilon = -10^{-4}$.
\label{historical_EW_ES_table}
}
\end{center}
\end{table}
}

{\small
\begin{table}[hbt!]
\begin{center}
{\small
\begin{tabular}{ccccc} \hline
${\kappa}$ &  $E[ \sum_i \qq_i]/T$  & ES (5\%) & $Median[W_T]$   \\ \hline
\hline
   .180 & 69.91 &-805.65 & -31.84\\
   1.0  & 61.77 &-290.03 & -40.87\\
   1.5  & 59.21 & -248.15& -77.26\\
   1.75 & 58.16 & -235.46& -78.50\\
   2.5  &6.02   &-219.00 &-81.84\\
   3.75 & 53.78 &  -209.9& -80.68\\
   5.0  &52.43  &  -207.15 & -77.25\\
   6.25 &51.74  &-209.02   & -78.11\\
   7.5  & 51.26 &-210.38   &-78.48\\
   10.0 & 50.58 & -212.41 &  -77.95\\
   100.0& 47.72 &-217.82  &  -67.91\\
 $\infty$ &  40.0 &-219.16  &      +17.34\\
\end{tabular}
}
\caption{EW-EW historical market results for optimal strategies,  assuming
the scenario given in Table~\ref{base_case_1}. No Tontine gains assumed.  Stock index: real capitalization weighted CRSP stocks;
bond index: real 30-day T-bills.  Parameters from Table \ref{fit_params}.
Units: thousands of dollars. Statistics
based on $10^6$  bootstrap simulation runs. Expected blocksize $=2$ years.
Control is computed using the Algorithm in Section \ref{DP_program} and Appendix \ref{Numerical_Appendix}, stored, and then used
in the bootstrap simulations.
$q_{\min} = 40$, $q_{\max} = 80$ (annually).
$T = 30 $ years,
$\epsilon = -10^{-4}$.
\label{historical_EW_ES_table_no_tontine}
}
\end{center}
\end{table}
}

\clearpage

\begin{singlespace}
\bibliographystyle{chicago}
\bibliography{paper}
\end{singlespace}

\end{document}